\documentclass[opre,sglanonrev]{informs4}
\usepackage{eqndefns-left} 
\RequirePackage{tgtermes}
\RequirePackage{newtxtext}
\RequirePackage{newtxmath}
\RequirePackage{bm}
\RequirePackage{endnotes}

\SingleSpacedXII 

\let\subparagraph\relax

\usepackage{balance}
\usepackage{amsfonts}
\usepackage{breakcites}
\usepackage{graphicx}
\usepackage{graphics,graphicx}
\usepackage{xcolor}
\usepackage{hyperref}
\usepackage{euscript}
\usepackage{multirow}
\usepackage{tikz}
\usepackage{algorithm,algpseudocode}
\usepackage{caption}
\usepackage{titlesec}
\usepackage{soul}
\usepackage{diagbox}
\usepackage{float}
\usepackage{comment}
\usepackage{cleveref}

\theoremstyle{TH}
\newtheorem{open}{Open problem}

\newcounter{alphasect}
\def\alphainsection{0}

\let\oldsection=\subsubsection
\def\subsubsection{%
	\ifnum\alphainsection=1%
	\addtocounter{alphasect}{1}
	\fi%
	\oldsection}%

\renewcommand\thesubsubsection{%
	\ifnum\alphainsection=1%
	Case \arabic{alphasect}%
	\else%
	\fi%
}%

\newenvironment{alphasection}{%
	\ifnum\alphainsection=1%
	\errhelp={Let other blocks end at the beginning of the next block.}
	\errmessage{Nested Alpha section not allowed}
	\fi%
	\setcounter{alphasect}{0}
	\def\alphainsection{1}
}{%
	\setcounter{alphasect}{0}
	\def\alphainsection{0}
}%

\makeatother

\titleformat{\paragraph}
{\normalfont\normalsize\bfseries}{\theparagraph}{1em}{}
\titlespacing*{\paragraph}
{0pt}{3.25ex plus 1ex minus .2ex}{1.5ex plus .2ex}

\titleformat{\subparagraph}
{\normalfont\normalsize\bfseries}{\thesubparagraph}{1em}{}
\titlespacing*{\subparagraph}
{0pt}{3.25ex plus 1ex minus .2ex}{1.5ex plus .2ex}

\newcommand{\fairwithsharing}[1]{\allowbreak {\sf F-Sharings}$(#1)$}
\newcommand{\propwithsharing}[1]{\allowbreak {\sf Prop-Sharings}$(#1)$}
\newcommand{\efwithsharing}[1]{\allowbreak {\sf EF-Sharings}$(#1)$}
\newcommand{\eqwithsharing}[1]{\allowbreak {\sf Eq-Sharings}$(#1)$}
\newcommand{\conswithsharing}[1]{\allowbreak {\sf Cons-Sharings}$(#1)$}

\newcommand{\propfpowithsharing}[1]{\allowbreak {\sf PropFpo-Sharings}$(#1)$}
\newcommand{\effpowithsharing}[1]{\allowbreak {\sf EFFpo-Sharings}$(#1)$}
\newcommand{\propfpowithshared}[1]{\allowbreak {\sf PropFpo-SharedObj}$(#1)$}
\newcommand{\effpowithshared}[1]{\allowbreak {\sf EFFpo-SharedObj}$(#1)$}
\newcommand{\propdpowithsharing}[1]{\allowbreak {\sf PropDpo-Sharings}$(#1)$}
\newcommand{\efdpowithsharing}[1]{\allowbreak {\sf EFDpo-Sharings}$(#1)$}
\newcommand{\propdpowithshared}[1]{\allowbreak {\sf PropDpo-SharedObj}$(#1)$}
\newcommand{\efdpowithshared}[1]{\allowbreak {\sf EFDpo-SharedObj}$(#1)$}

\newcommand{\fairwithshared}[1]{\allowbreak {\sf F-SharedObj}$(#1)$}
\newcommand{\propwithshared}[1]{\allowbreak {\sf Prop-SharedObj}$(#1)$}
\newcommand{\efwithshared}[1]{\allowbreak {\sf EF-SharedObj}$(#1)$}
\newcommand{\eqwithshared}[1]{\allowbreak {\sf Eq-SharedObj}$(#1)$}
\newcommand{\conswithshared}[1]{\allowbreak {\sf Cons-SharedObj}$(#1)$}
\newcommand{\maxminnRwayidenticalsplitpar}[1]{\allowbreak {\sf MaxMinIdentical}$(#1)$}

\newcommand{\sumgb}[1]{U_{#1}}

\newcommand{\citet}[1]{\cite{#1}}
\newcommand{\citep}[1]{\cite{#1}}

\newcommand{\range}[2]{\in\{#1,\dots,#2\}}

\definecolor{ForestGreen}{rgb}{.13,.54,.13}
\definecolor{BrickRed}{rgb}{.80,.26,.33}

\newcommand{\hard}[1]{\textcolor{BrickRed}{NP-hard [#1]}}
\newcommand{\stronglyhard}[1]{\textcolor{BrickRed}{Strong NP-hard [#1]}}
\newcommand{\complete}[1]{\textcolor{BrickRed}{NP-complete [#1]}}
\newcommand{\easy}[1]{\textcolor{ForestGreen}{#1}}

\usepackage{natbib}
 \bibpunct[, ]{(}{)}{,}{a}{}{,}%
 \def\bibfont{\small}%
\EquationsNumberedThrough    

\TheoremsNumberedThrough     
\ECRepeatTheorems  %

\MANUSCRIPTNO{}
\begin{document}

\RUNAUTHOR{S. Bismuth et al.}

\RUNTITLE{Fair Division with Bounded Sharing: Binary and Non-Degenerate Valuations}

\TITLE{Fair Division with Bounded Sharing:\\Binary and Non-Degenerate Valuations\footnote{This is an extension of a paper presented in a conference [details redacted for anonymity but are given in the cover letter].
The current version is significantly expanded, adding new results: 
\Cref{thm:cons-np-hard}, 
\Cref{thm:unbounded-binary-ef}, 
\Cref{thm:unbounded-binary-prop-no-sharings},
\Cref{thm:unbounded-binary-eq-no-sharings}
\Cref{thm:unbounded-binary-eq-sharing}, \Cref{thm:unbounded-binary-eq-shared-object}, as well as an extended treatment of the new notions of generically polynomial-time algorithms and generic NP-hardness, and all proofs
omitted from the conference version.}}

\ARTICLEAUTHORS{%
\AUTHOR{Ivan Bliznets}
\AFF{Department of Computer Science,
Ariel University, Ariel, 40700, Israel, \EMAIL{samuelbismuth101@gmail.com}}

\AUTHOR{Samuel Bismuth}
\AFF{Department of Computer Science,
University of Groningen, Groningen, 9747 AG, Netherlands, \EMAIL{todo@gmail.com}}

\AUTHOR{Erel Segal-Halevi}
\AFF{Department of Computer Science,
Ariel University, Ariel, 40700, Israel, \EMAIL{todo@gmail.com}}
} 

\ABSTRACT{%
A set of objects is to be divided fairly among agents with different tastes, modeled by additive utility-functions. If we consider the objects as indivisible, many instances of the decision problem: ``Is there a fair division of the objects among the agents'' are negative. In addition, this question is hard to solve even for most of the special cases. The latter reasons give us a good motivation to relax the problem for which the running time complexity is better, and the number of positive instances (admitting a fair division) will significantly grow. Whereas many works relax the fairness criteria, this paper introduces another relaxation: an agent is allowed to share a \emph{bounded} number of objects between two or more agents in order to attain fairness. 

The paper studies various notions of fairness, such as proportionality, envy-freeness, equitability, and consensus.
We analyze the run-time complexity of finding a fair allocation with a given number of sharings under several restrictions on the agents' valuations, such as: binary, generalized-binary, and non-degenerate.
}%




\maketitle
\

\section{Introduction}\label{sec:intro}

Fair division with divisible and/or indivisible objects is a fundamental problem in operations research, and has been extensively studied across diverse contexts and applications \citep{sandomirskiy2022efficient,DBLP:journals/mansci/BogomolnaiaMS22,DBLP:journals/mor/BogomolnaiaM23,DBLP:journals/mor/BranzeiS24,DBLP:journals/eor/AzizHMS23,DBLP:journals/eor/CornillyPRV22}.

Consider several siblings who have inherited some assets and need to decide how to allocate them, 
or several parties who form a coalitional government and need to allocate the government ministries, 
or several faculty members who have moved to a new building and need to allocate the offices. 
In all these cases, a set of valuable objects has to be allocated among several agents, who may have different preferences over the objects, and it is important that all agents view the allocation as \emph{fair}.

In many cases, fairness can only be attained by giving fractions of the same object to different agents. For example, if three siblings inherit four identical houses, then fairness requires that each sibling receives one house plus $1/3$ of the fourth house.
Fractional allocation means that some objects must be \emph{shared} among two or more agents. 

Sharing can be implemented in various ways. For example, sharing a house can be implemented by renovating it in a way that will enable all three siblings to live in it simultaneously;
sharing a cabinet ministry is often done by a rotational agreement, in which each party controls the ministry for a fraction of the time. Still, sharing an object is inconvenient, so it is desirable to share as few objects as possible. In the above example, a fair allocation could also be attained by sharing all four houses, giving each of the three siblings $1/3$ of each house, but this allocation is clearly less desirable than the allocation in which only one house is shared.
This motivates the following generic computational problem, which is at the heart of the present research:
\begin{quote}
(*)
\emph{Given $m$ objects, $n$ agents with different valuations over the objects, a fairness notion, and an integer $s\geq 0$, find a fair allocation in which at most $s$ objects are shared, if such an allocation exists.}
\end{quote}
We also consider a variant in which $s$ is the number of \emph{sharings} rather than the number of shared objects (e.g. a single object shared between $10$ agents counts as $9$ sharings).
We assume that agents have linear additive valuations, and that all objects are goods (have non-negative values).
We consider four common fairness notions: 
\emph{proportionality} (each agent values his share as at least $1/n$ of the total value), \emph{envy-freeness} (each agent values his share as at least the share of any other agent), \emph{equitability} (the subjective value of all agents is the same), and \emph{consensus} (all $n$ agents value all $n$ bundles exactly the same).
 see \textbf{\Cref{sec:model}} for the formal definitions. 

The case $s=0$ of (*) is strongly NP-hard, as it can be used to solve the strongly NP-hard problem \textsc{3-Partition}. 
But many real-life problems (e.g. inheritance allocation or cabinet ministries) involve a small number of agents, so there is value in studying the case when $n$ is a small fixed constant.
Even in that case the problem is NP-hard in general, as the case of $s=0$ and $n=2$ agents with identical valuations is equivalent to the NP-hard \textsc{Partition} problem.

At the other extreme, when $s$ is sufficiently large the problem becomes polynomial-time solvable. In particular, when $s=n-1$ it is always possible to find in polynomial time a proportional allocation, an envy-free allocation and an equitable allocation; when $s=n(n-1)$ 
the same holds for a consensus allocation (these results are known from previous work; in \textbf{\Cref{sec:arbitrary}} we provide simple proofs using linear programming).

Experiments on real-life and simulated instances show that most instances admit a fair allocation with fewer sharings than the worst case \citep{DBLP:journals/corr/abs-2204-11753,sandomirskiy2022efficient}.
This motivates the question of deciding whether a specific instance admits a fair allocation with $s$ shared objects or sharings, where $s$ is smaller than the upper bound. In our opening example, as there are $n=3$ siblings, the worst-case upper bound is $n-1=2$, but there exists an allocation in which only $s=1$ house is shared, which is more convenient than having to share two houses.
This challenge is the focus of the present paper. For various settings and classes of valuation functions, we aim to determine the values of $s$ for which (*) can be solved in polynomail time.


\textbf{\Cref{sec:binary-utilities}} presents new results on the \emph{binary valuations} case (every agent values every object at 0 or 1). This case was not handled before, but it is easy to show that, for any constant $n$ and $s$, the existence of a fair allocation for $n$ agents with $s$ sharings / shared objects can be decided in polynomial time by a mixed integer linear program with a fixed number of variables. Some results for the case where $n$ is unbounded are also given in this section.

\textbf{\Cref{sec:generalized-binary}} considers agents with \emph{generalized binary} valuations \citep{DBLP:journals/mss/CamachoFPT23},
also known as \emph{cost valuations} \citep{DBLP:conf/sagt/BotanRSW23}  (for every object $o$ there is a price $p_o$ such that each agent $i$ values $o$ at either $p_o$ or $0$). 
Generalized binary valuations were introduced in the context of the \emph{Santa Claus problem} \citep{bansal2006santa}.
We focus on a subset of the generalized binary valuations in which 
the sum of the utilities is equal for each agent. 
We call these  \emph{equal-sum generalized binary} valuations.
These valuations generalize identical valuations that were studied in \citet{DBLP:journals/corr/abs-2204-11753}; the results there imply that deciding existence of fair allocation with $s\leq n-3$ shared objects, or with $s\leq n-2$ sharings, are both NP-hard for any fixed $n\geq 3$.
The remaining unsolved cases for equal-sum generalized binary valuations are the cases with $s=n-2$ shared objects.
We present a polynomial-time algorithm for deciding the existence of a proportional allocation for $n=3$ agents and $s=1$ shared object.
We note that, in fair item allocation, even the case of three agents is often interesting and non-trivial. For example, an important recent paper \citep{chaudhury2024efx} is devoted to finding an EFX allocation for three agents, whereas the case of four agents is still open.

\textbf{\Cref{sec:non-degenerate}} studies agents with \emph{non-degenerate} valuations (for every two agents, there are no two objects such that the value-ratios are equal). 
Non-degeneracy is arguably a weak requirement, as, informally, almost all valuations are non-degenerate
 \citep{sandomirskiy2022efficient} 
 (see formal statement in
 \Cref{prop:generically-polynomial}).
Polynomial-time solvability for non-degenerate valuations means that almost all instances are ``easy''; this result is in the spirit of smoothed analysis \citep{spielman2005smoothed,Moitra_2011}.
But fPO --- the other assumption made by \citet{sandomirskiy2022efficient} --- is a strong requirement,
as, informally, ``almost all'' allocations are not fPO.%
\footnote{
Formally, even without sharing, the number of fPO allocations is in $O(m^{{n \choose 2} + 2})$ \citep{sandomirskiy2022efficient}, whereas the total number of allocations is at least $n^m$,
so the fraction of fPO allocations goes to $0$ as $m\to\infty$, when $n$ is fixed.
}
Dropping the fPO requirement may enable allocation with fewer sharings, which the agents may prefer to an fPO allocation with many sharings.
This raises the question of whether a fair allocation (not necessarily fPO) can be found efficiently for agents with non-degenerate valuations.
Our findings are mostly negative: for most fairness notions, we prove NP-hardness even for non-degenerate valuations.

Our proofs use an unusual reduction technique, which may be of independent interest.
Usually, NP-hardness of a problem $P_2$ is shown by reduction from a known NP-hard problem $P_1$, where each instance of $P_1$ is transformed to a \emph{single} instance of $P_2$ with the same answer. 
We define \emph{multi-reductions}, in which, for each instance of $P_1$, we construct a large set of instances of $P_2$ with the same answer.
We use these multi-reduction to prove that (unless P=NP) it is \emph{not} true that ``almost all instances of $P_2$ can be solved in polynomial time''.

\textbf{\Cref{sec:truthful}} studies \emph{truthful mechanisms} for fair allocation.
A truthful mechanism is an algorithm which incentivizes the agents to reveal their true valuations. 
Truthfulness is another reason to drop the fPO requirement: it is known that no truthful mechanism can guarantee both fairness and Pareto-efficiency; it may be desired to give up efficiency to get truthfulness. We survey several truthful fair allocation algorithms, and check whether they can be adapted to construct an allocation with bounded sharing. 

The results from this and previous works are summarized and compared in \Cref{tab:nagents}.

\begin{table*}[h!]
\small
\center
\begin{tabular}{|l|l|l|l|l|l|}
\hline
\textbf{Valuations} & \textbf{Allocation} & \textbf{Nb of agents}  & \textbf{Measure} & \textbf{bound} & \textbf{Run-time complexity} \\ \hline \hline \hline

\textbf{Identical} & Fair & Unbounded & sharing & $s$ (any) & \stronglyhard{1} \\ \cline{4-6} 
                   & (with  & & shared object & $s$ (any) & \stronglyhard{1} \\ \cline{3-6} \cline{3-6}
                   & identical & Constant $n$ & sharing & $s \leq n - 2$ & \complete{1} \\ \cline{5-6}
                   & valuations,  & & & $s \geq n - 1 $ & \easy{$O(m+n)$ [cut-the-line]} \\ \cline{4-6} 
                   & all fairness  & & shared & $s \leq n - 3 $ & \complete{1} \\ \cline{5-6}
                   & concepts & & & $s = n - 2 $ & \easy{$O(poly(m, \log (V_1)))$ [1]} \\ \cline{5-6}
                   & coincide.) & & & $s \geq n - 1 $ & \easy{$O(m + n)$ [cut-the-line]} \\ \hline \hline



 \textbf{Arbitrary} & PROP & Constant $n$ & both & $s \geq n-1$ & \easy{LP [\Cref{thm:upperbound-prop}]}  \\
 & &  & &  & \easy{$O(mn\log(n))$ \citep{Even1984Note}}
 \\ \cline{2-2}  \cline{4-6} 

 & EF &  & both & $s \geq n-1$ &  \easy{LP [\Cref{thm:upperbound-ef}]}  \\
& &  & &  & \easy{$O((n+m)^4\log(n+m))$ \citep{orlin2010improved}} \\ \cline{2-2}  \cline{4-6} 


& EQ &  & both & $s \geq n-1$ &\easy{LP [\Cref{thm:upperbound-eq}]}   \\ \cline{2-2}  \cline{4-6} 

& CONS &  & both & $s \geq n (n-1)$ & \easy{LP [\Cref{thm:upperbound-cons}]} \\  
&  &  &  & & \easy{Strong poly \cite{goldberg2020consensus}}  
\\ \cline{4-6} 
&  &  & both & $s \leq n-1$ & \hard{\Cref{thm:cons-np-hard}}  \\ \cline{2-2}  \cline{4-6}


\hline \hline

\textbf{Binary} & EF & Unbounded & both & $s$ (any) & \complete{\Cref{thm:unbounded-binary-ef}} \\ \cline{2-2} \cline{5-6} 
                   & PROP &  &  & $s = 0$ & \easy{$O(poly(m,n))$[\Cref{thm:unbounded-binary-prop-no-sharings}]} \\ \cline{2-2} \cline{6-6} 
                   & EQ & &  &  & \easy{$O(poly(m,n))$[\Cref{thm:unbounded-binary-eq-no-sharings}]} \\ \cline{4-6} 
                    & & & sharing & $s$ (any) & \easy{$O(poly(m,n))$[\Cref{thm:unbounded-binary-eq-sharing}]} \\ \cline{4-6} 
                    & & & shared object & $s$ (any) & \easy{$O(poly(m,n))$[\Cref{thm:unbounded-binary-eq-shared-object}]} \\ \cline{2-6} 
                   & All & Constant $n$ & both & $s$ (any) & \easy{MILP [\Cref{thm:milp-binary}]}
                   \\ \hline \hline 
                 
\textbf{Equal-} & Fair  & Constant $n$ & shared object & $s \leq n-3$ & \complete{see Identical} \\ \cline{5-6} 
\textbf{sum-}&   &  & & $s \geq n-1$ & \easy{Weak poly [see Arbitrary]} \\  \cline{2-6}
\textbf{generalized} & PROP  & 3 & shared object & $s = 1$ & \easy{$O(poly(m, \log (V_1))$ [\Cref{thm:g-binary}]} \\ \cline{2-6}

 & All  & Constant $n$ & sharing & $s \leq n-2$ & \complete{see Identical} 
 \\ \cline{5-6}
 &   &  & & $s \geq n-1$ & \easy{Weak poly [see Arbitrary]} \\ 
\hline 
\hline 

\textbf{Non-} & PROP, EF & Unbounded & sharing & $s$ (any) & \stronglyhard{\Cref{thm:unbounded-n-sharing-nondegenerate}} \\ \cline{4-6}
\textbf{degenerate} &  &  & shared object &$s$ (any) & \stronglyhard{\Cref{thm:shared-strong-np-degenerate}} \\
\cline{2-6}

& PROP, EF & Constant $n$ & sharing & $n\geq 2, s=0$& \complete{\Cref{thm:hardness-fair-with-sharings}} \\ \cline{4-4} \cline{6-6}
& &  & & or $n\geq 3,$ & \\
& &  & shared object & $s\leq n-3$ & \complete{\Cref{thm:hardness-fair-with-shared}} \\  \cline{2-2} \cline{4-6}
& PROP + dPO & & both & $s=0$ & \complete{\Cref{thm:hardness-dpo-with-sharings}} \\ \cline{2-2}
& EF + dPO & &  &  &  \\  \cline{2-2} \cline{5-6}


& Fair + fPO &  &  & $s \geq n-1$ & \easy{$O(m^{poly(n)})$ [2]}  \\ 
\hline

\end{tabular} 

\caption{\label{tab:nagents} 
Run-time complexity of allocating $m$ objects among agents, with a bound on the number of sharings/shared objects. \\
$[1]$ \citep{DBLP:journals/corr/abs-2204-11753} \\
$[2]$ \citep{sandomirskiy2022efficient}
}
\end{table*}

\section{Related work}
The impossibility of achieving fair allocation with indivisible objects has prompted numerous researchers to adopt relaxed fairness criteria and pursue approximate fairness guarantees.

\subsection{Relaxations with indivisible objects}
A common weakening of envy-freeness is 
\emph{Envy-freeness except one good (EF1)}. EF1 was introduced by \citet{budish2011combinatorial}, and a similar idea appeared earlier in \citet{Lipton2004Approximately}.
EF1 has been widely explored, for instance by \citet{Aleksandrov2015Online,oh2018fairly} and others.
A stronger version of EF1, aiming for a global fairness notion through information hiding, was studied by \citet{hosseini2019fair,BliznetsBS24}.
The existence of EF1 allocations that are also Pareto Optimal was shown by \citet{caragiannis2016unreasonable}, and \citet{barman2018finding} extended this to fractional PO.

A common weakening of proportionality is the \emph{Maximin share guarantee (MMS)}. MMS was defined by \citet{budish2011combinatorial}, who showed its existence under a large-market assumption and applied it in real-world course allocation \citep{budish2016course}.
In ``small markets'', \citet{Procaccia2014Fair} showed that MMS allocations might not always exist in certain edge cases. Because of this, most MMS-related work focuses on approximate versions of the guarantee, such as in \citep{amanatidis2017approximation, barman2017approximation, ghodsi2018fair, garg2018approximating, aziz2016approximation, Babaioff2017Competitive, segal2020competitive}.

Traditional microeconomics mainly deals with divisible resources and treats indivisible goods by turning them into ``divisible'' ones through lotteries. This method leads to a weaker form of fairness: allocations are fair ex-ante, meaning in expectation before the lottery takes place. Ex-ante fairness has been studied in various settings. For example, :  \citet{hylland1979efficient}, \citet{abdulkadiroglu1998random},  \citet{Bogomolnaia2001New} looked at fair assignment problems; \citet{budish2013designing} considered its multi-unit constrained modifications; \citet{kesten2015theory} evaluated fairness of tie-breaking in matching markets; \citet{bogomolnaia2019simple} investigated randomized rules for online fair division.

\citet{Brams2013TwoPerson} observed that exact envy-freeness with indivisible goods can be reached by leaving some goods unallocated, while still maintaining certain efficiency guarantees. Their AL method produces an allocation that isn’t Pareto-dominated by any other envy-free allocation (see also \citet{Aziz2015Generalization}).
In a later work \citep{brams2014algorithm}, the same authors have presented an algorithm that is \emph{maximally-proportional}, in that it assigns a proportional bundle to as many agents as possible, while possibly leaving some items unassigned.
Later, \citet{caragiannis2019envy} proved that by setting aside some items, it’s possible to find an allocation satisfying ``Envy-free except any good'' (EFX), a fairness criterion that strengthens EF1 \citep{caragiannis2016unreasonable}.


\subsection{Relaxations with divisible objects}

Most traditional works on fair division either assume that all objects are divisible, or that all objects are indivisible. 
But in recent years, there is a growing interest in ``mixed'' settings; see \citet{liu2024mixed} for a survey.

The works closely related to ours are:
\begin{itemize}
\item 
\citet{bei2021fair} and 
\citet{bei2021maximin} study an allocation problem where some goods are divisible and some are indivisible. In their setting, the difference between divisible and indivisible objects is predefined --- the algorithm is only allowed to divide the objects that are already marked as divisible. This means that an exact-fair allocation may not be possible, so they study approximate-fairness notions. 
\item 
\citet{bhaskar2021approximate} study a mixture of indivisible and divisible chores (objects with negative utilities), and present an algorithm for approximate envy-free allocation.
\item 
\citet{nishimura2023envy} and 
\cite{kawase2023fair} study welfare-maximization problems in settings of mixed goods.
\item 
\citet{bei2023fair} generalize the above model by allowing \emph{subjective divisibility}, wherein each object can be considered divisible by some agents and indivisible by others. Again, exact-fair allocation may not be possible, but they show algorithms for approximate envy-free and maximin-share allocations, focusing mainly on two and three agents.
\item 
\citet{li2024allocating} study a setting in which each agent may have a different ``indivisibility ratio'' (= proportion of items that are indivisible), and the approximate-fairness guarantee for each agent depends on his indivisibility ratio.
This is another way to allow a bounded number of divisible objects. While they bound the number of divisible objects per agent, we study a more global bound in the number of objects.
\item 
\citet{DBLP:conf/atal/LiLLTT24} study the price of fairness in both indivisible and mixed item allocation. 
\end{itemize}
A crucial difference between these works and our paper is that we insist on \emph{exact} fairness. This is important in settings with highly valuable objects, in which approximate fairness may be unacceptable.

\subsection{Exact fairness with divisible objects}

Exact fair allocation with bounded sharing was first studied by 
\citet{Brams1996Fair,brams2000winwin}.
Their \emph{Adjusted Winner (AW)} procedure finds an allocation for $n=2$ agents with at most $s=1$ shared object, that is simultaneously proportional, envy-free, equitable, and \emph{fractionally Pareto-optimal} (fPO: no other fractional allocation is at least as good for all agents and strictly better for some agent).
They also show an example with three agents \citep{Brams1996Fair} where no allocation is simultaneously fPO, envy-free and equitable. This does not rule out the option of satisfying each of these properties on its own.
The AW procedure has been (at least in theory) used in division problems involving divorce cases and international conflicts \citep{Brams1996Camp,Massoud2000Fair}, and it has also been examined through empirical studies \citep{Schneider2004Limitations,Daniel2005Fair}.

For $n\geq 3$ agents, the number of required sharings was studied in an unpublished manuscript of 
\citet{wilson1998fair}. 
He proved the existence of an \emph{egalitarian}
allocation of goods --- an allocation in which all agents have the largest possible equal utility \citep{pazner1978egalitarian} --- with $n-1$ sharings. Egalitarian allocations are proportional but not necessarily envy-free.

Several more recent works have proved a polynomial upper bound on the required number of sharings for other fairness notions (\Cref{sec:bounds} provides complete proofs and references):
\begin{itemize}
\item For proportionality, envy-freeness and equitability, there always exists a fair allocation with $s=n-1$ (sharings or shared objects), and there may not exist a fair allocation with smaller $s$.
\item There always exists a consensus allocation with $s=(n-1)n$ (sharings or shared objects), and there may not exist a fair allocation with smaller $s$.
\end{itemize}
In all cases, an allocation satisfying the worst-case upper bound can be computed in polynomial time (see \Cref{sec:bounds}).

However, several recent works provide polynomial-time algorithms for some other special cases:
\begin{enumerate}
\item 
\citep{DBLP:journals/corr/abs-2204-11753}
consider agents with \emph{identical} valuations. With identical valuations, all fairness notions coincide, and are equivalent to finding a perfect scheduling of $m$ jobs on $n$ identical machines.%
\footnote{
In fact, their result pertains also to the more general case of \emph{uniform machines}, which is equivalent to agents having identical valuations but different entitlements.
}
They develop a polytime algorithm for deciding if there exists a fair allocation with $s=n-2$ shared objects, that is, one fewer than the worst-case upper bound, for any fixed $n\geq 3$. They prove that the $n-2$ is tight, as the problem is NP-hard for any fixed $n\geq 3$ and $s\leq n-3$.
\item 
\citet{sandomirskiy2022efficient} go to the other extreme and consider agents with \emph{non-degenerate} valuations (for every two agents, their value ratios for the $m$ objects are all different).
They also require the allocation to be fPO in addition to being fair.
They prove that, with non-degenerate valuations, the number of fPO allocations with $s$ shared objects is polynomial in $m$ (for every fixed $n$), so it is possible to enumerate all such allocations in polynomial time and check whether one of them satisfies any desired fairness notion.
Therefore, the problem (*) is in P for every fixed $n$.
\cite{misra2021fair} complement this result by proving that, when $n$ is not fixed, the problem is NP-hard even for non-degenerate valuations and $s=0$.
\item 
\citep{goldberg2022consensus} 
study \emph{consensus splitting} --- a partition of  objects into $k$ subsets each of which has a value of exactly $1/k$ for all agents.
They show that computing a partition with at most $(k-1)n$ sharings can be done in polynomial time. However, computing a partition with fewer sharings is hard even for $k=2$: for any fixed $n$ and any $s<n$, it is NP-hard to decide whether a partition with at most $s$ sharings exists.
\end{enumerate}



\section{Preliminaries}
\label{sec:model}
\subsection{Agents, objects, and allocations}
There is a set $[n] = \{1,\ldots,n\}$ of agents and a set $[m]=\{1,\ldots,m\}$ of objects.
For each agent $i\in[n]$ and object $o\in[m]$, the value $v_{i,o}\in \mathbb{Q}$ represents agent $i$'s utility of receiving the object $o\in[m]$ in its entirety. 
The set of all instances is $\mathbb{Q}^{n m}$, representing the set of all $n\times m$ matrices.
~
The total value of agent $i$ to all objects is denoted by $V_i := \sum_{o\in [m]} v_{i,o}$.
In general, the matrix $\mathbf{v}$ may contain values of mixed signs; but in this paper, we focus on allocation of goods and assume that all elements of $\mathbf{v}$ are non-negative.

A \emph{bundle} $\mathbf{x}$ of objects is a vector $(x_o)_{o\in [m]}\in[0,1]^m$, where the component $x_o$ represents the fraction of object $o$ in the bundle. 
Each agent $i\in[n]$ has a \emph{utility function} $u_i$, assigning a numeric utility to each bundle.
The utility functions are assumed to be \emph{linear} and \emph{additive}, which means that $
u_i(\mathbf{x}) =
\sum_{o\in [m]}  v_{i,o}\cdot x_{o}
$. 

An \emph{allocation} $\mathbf{z}$ is a collection of bundles $(\mathbf{z_i})_{i\in [n]}$, one for each agent, with the condition that all the objects are fully allocated. An allocation can be identified with the matrix $\mathbf{z} := (z_{i,o})_{i\in[n],o\in[m]}$  such that all $z_{i,o}\geq 0$ and $\sum_{i\in [n]} z_{i,o} = 1$ for each $o\in[m]$.

\subsection{Fairness and efficiency concepts}
\label{sub:fairness-notions}
We focus on four common fairness concepts. 
An allocation $\mathbf{z}$  is called:

- \emph{Proportional (PROP)} --- if every agent prefers his bundle to the equal division. Formally, for all $i\in[n]$:~~
$u_i(\mathbf{z}_i) \geq V_i/n$. 

- \emph{Envy-free (EF)} --- if every agent prefers his bundle to the bundles of others. Formally, for all $i,j\in[n]$:~~ $u_i(\mathbf{z}_i) \geq u_i(\mathbf{z}_j)$.
Every envy-free allocation is also proportional; with $n=2$ agents, envy-freeness and proportionality are equivalent.

- \emph{Equitable (EQ)} --- if it gives each agent exactly the same value. Formally, for all 
$i,j\in[n]$: $u_i(\mathbf{z}_i)= u_j(\mathbf{z}_j)$.

- \emph{Consensus (CONS)} --- if every agent attributes exactly the same  value to every bundle:
for all $i,j\in[n]$: 
$u_i(\mathbf{z}_j) \equiv U$ for all $i,j\in N$, for some constant $U$.
A consensus allocation is proportional,  envy-free and equitable.
Often a more general notion is considered, in which the number of parts may be different than $n$; 
a \emph{consensus $k$-partition}
is a partition of the objects into $k$ bundles satisfying the consensus condition.
Note that a consensus allocation exists only if the total value is the same for all agents, $V_i \equiv k \cdot U$.

We also consider two common efficiency concepts. 
We say that an allocation $\mathbf{z}$ is \emph{Pareto-dominated} by an allocation $\mathbf{y}$ if $\mathbf{y}$ gives at least the same utility to all agents and strictly more to at least one of them.
An allocation $\mathbf{z}$ is called:

\noindent
- \emph{Fractionally Pareto-optimal (fPO)}: 
if it is not dominated by any  allocation $\mathbf{y}$.

\noindent
- \emph{Discretely Pareto-optimal (dPO)}:
if it is not dominated by any allocation $\mathbf{y}$ with no sharing.

\subsection{Measures of sharing}\label{sub:measures-of-sharing}
If for some $i\in [n]$, $z_{i,o} = 1$, then the object $o$ is not shared --- it is fully allocated to agent $i$. 
Otherwise,  object $o$ is shared between two or more agents.
Throughout the paper, we consider two measures  quantifying the amount of sharing in a given allocation $\mathbf{z}.$

The simplest one is \emph{the number of shared objects}:
\begin{align*}
\big|\left\{o\in[m]\, :\, z_{i,o}\in(0,1)\mbox{ for some }i\in[n]\right\}\big|.
\end{align*}
Alternatively, one can take into account the number of times each object is shared. This is captured by \emph{the number of sharings}
\begin{align*}
\sum_{o\in[m]}\bigg(\big|\{i\in [n]:\, z_{i,o}>0\}\big|-1\bigg) 
\end{align*}
Both measures are zero for discrete allocations. They differ, for example, if only one object $o$ is shared but each agent consumes a bit of $o$: the number of shared objects in this case is $1$ while the number of sharings is $n-1$.
Clearly, the number of shared objects is smaller than the number of sharings
for every allocation.

\subsection{Types of utilities}

Recall that for each agent $i \in [n]$ and object $o \in [m]$, the value $v_{i,o}\in \mathbb{Q}$ represents agent $i$'s utility of receiving the object $o\in[m]$ in its entirety. In general, $v_{i,o}$ can be any value in $\mathbb{Q}$ with no relation with other agent values. We relate to this as \emph{arbitrary valuations}. 
We consider several special classes:
\begin{itemize}
\item \emph{Identical valuations} --- for each object $o \in [m]$, there is a rational number $p_o>0$ such that $v_{i,o} = p_o$ for every agent $i \in [n]$.
\item \emph{Binary valuations} --- Each agent $i \in [n]$ values every object $o \in [m]$ as either $v_{i,o} = 0$ or $v_{i,o} = 1$.
\item \emph{Generalized binary valuations} 
--- for each object $o \in [m]$, there is a rational number $p_o > 0$ such that, for every agent $i \in [n]$,
either $v_{i,o} = p_o$  or $v_{i,o} = 0$.%
\footnote{A similar utility function is used in \cite{DBLP:conf/sagt/BotanRSW23, DBLP:journals/mss/CamachoFPT23}.
}
In the special case \emph{Equal-sum generalized binary valuations},  the sum of values of all objects is the same for all agents.
\item \emph{Non-degenerate valuations} --- 
for every two agents $i,j \in [n]$, there are no two objects $o_1,o_2 \in[m]$ such that the value-ratios are equal ($v_{i,o_1} \cdot v_{j,o_2} = v_{i,o_2} \cdot v_{j,o_1}$).
\end{itemize}

\subsection{Computational problems}
In this paper we consider computational problems of the following kinds.

\begin{definition}
Let F be a fairness criterion (proportionality, envy-freeness, etc.). 

(a) For any fixed integers $n\geq 2$ and $s\geq 0$, 
\fairwithsharing{n,s} is the problem of deciding if a given instance with $n$ agents admits an F allocation with at most $s$ sharings.
\fairwithshared{n,s}  is the same problem where $s$ is the maximum number of shared objects.

(b) 
For any fixed integer $s\geq 0$,
\fairwithsharing{s} is the problem of deciding if a given instance admits an F allocation among $n$ agents (where $n$ is part of the input) with at most $s$ sharings;
\fairwithshared{s} is the same problem where $s$ is the maximum number of shared objects.
\end{definition}
All those problems are obviously in NP, so proving NP-hardness is enough to prove that those problems are NP-complete.

\section{General Additive Valuations}
\label{sec:arbitrary}\label{sec:bounds}
This section considers the general case in which agents may have arbitrary additive valuations.

\subsection{Upper bounds using linear programming}
\label{sub:LP}
In this section, 
we show that for every fairness criterion F among our four criteria, there is a function $S_F(n)$ such that any instance with $n$ agents has an ``F'' allocation with at most $S_F(n)$ sharings.
Most of the results below were proved previously using combinatorial techniques \citep{bogomolnaia2016dividing,barman2018proximity,sandomirskiy2022efficient}.
We present below alternative proofs for all these results using linear programming. The advantage of this technique is that it is more general, and can potentially be used to derive upper bounds for other fairness notion.

The main tool is the following known fact: a bounded feasible linear program 
in equational form, 
with $k$ 
equational 
constraints, has a \emph{basic feasible solution} --- a solution in which at most $k$ variables are non-zero \citep{matousek2007understanding}.
The technique was first applied to fair division by \citet{wilson1998fair}; below we generalize his results.

\begin{lemma}
\label{thm:upperbound}
For any allocation $\mathbf{z^*}$:

(a) There exists an allocation $\mathbf{z}$ with at most $n$ sharings 
in which each agent receives exactly the same utility:
$\forall i\in[n]: u_i(\mathbf{z}_i) = u_i(\mathbf{z^*}_i)$.
The $n$ is tight even for shared objects.

(b) There exists an allocation $\mathbf{z}$ with at most $n-1$ sharings 
in which each agent receives at least as much utility:
$\forall i\in[n]: u_i(\mathbf{z}_i) \geq u_i(\mathbf{z^*}_i)$.
The $n-1$ is tight even for shared objects.
\end{lemma}
\proof{Proof}
For each $i\in[n]$, define the constant $U_i := u_i(z^*_i)$.

(a) 
The allocation $\mathbf{z}$ can be found using the following LP,
in which the variables are the $z_{io}$ for all $i\in[n],o\in[m]$, and the constraints are:
\begin{align}
\label{eq:m+n}
&& \sum_{i=1}^n z_{io} = 1 && \text{for } o\in[m]
\\
\notag
&& \sum_{o=1}^m v_{io} z_{io} = U_i && \text{for } i\in[n]
\\
\notag
&& 
z_{io}\geq 0 && \text{for } i\in[n], o\in [m]
\end{align}

This LP has at least one solution ($\mathbf{z^*}$) and has $m+n$ equational constraints, so it has a solution $\mathbf{z}$ in which at most $m+n$ variables are non-zero: $\sum_{o\in[m]}\left(\big|\{i\in [n]:\, z_{i,o}>0\}\big|\right)\leq m+n$.
By definition of sharing, the number of sharings is smaller or equal than $(m+n)-m=n$.

For tightness, consider an instance with $n$ agents and $n$ goods, where agent $i$ values good $i$ at $n^2-n+1$ and the other $n-1$ goods at $1$. Let $\mathbf{z}^*$ be the allocation giving each agent $1/n$ of each good. This allocation has $n$ sharings and $n$ shared objects, and gives each agent a utility of exactly $n$. But no allocation with less than $n$ shared objects can attain exactly the same utilities: if object $i$ is given entirely to agent $i$, then agent $i$'s utility is more than $n$; if it is given entirely to another agent, then agent $i$'s utility is at most $n-1$.

(b) The above LP can be modified by removing the constraint for agent $n$, and maximizing agent $n$'s utility instead:

\begin{align}
\label{eq:m+n-1}
\text{maximize} && \sum_{o=1}^m v_{no} z_{no}
\\
\notag
\text{subject to}
&& \sum_{i=1}^n z_{io} = 1 && \text{for } o\in[m]
\\
\notag
&& \sum_{o=1}^m v_{io} z_{io} = U_i && \text{for } i\in[n-1]
\\
\notag
&& 
z_{io}\geq 0 && \text{for } i\in[n], o\in [m]
\end{align}
There is a feasible solution 
($\mathbf{z^*}$) in which the objective value is $U_n$, so in every optimal solution the utility of agent $n$ is at least $U_n$ and the utility of each agent $i<n$ is exactly $U_i$.
The LP has $m+n-1$ equational constraints, so it has a solution $\mathbf{z}$ with at most $m+n-1$ non-zeros. In this solution the number of sharings is smaller or equal than $(m+n-1)-m=n-1$.

For tightness, consider an instance with $n$ agents and $n-1$ goods whose value for all agents is $n$. Let $\mathbf{z^*}$ be the allocation that gives each agent $1/n$ of each good.
This allocation has $n-1$ sharings and gives each agent a utility of $n-1$. But in any allocation with less than $n-1$ shared objects, at least one object is given entirely to some agent, and thus the remaining $n-1$ agents must share a total value of at most $(n-2)n < (n-1)(n-1)$; thus at least one remaining agent gets less than $(n-1)$.
\Halmos \endproof

As a corollary, we have two upper bounds (some of them proved in different ways by \citet{sandomirskiy2022efficient}).
\begin{theorem}
\label{thm:upperbound-prop}
Any instance with $n$ agents admits a proportional allocation with at most $n-1$ sharings. The $n-1$ is tight even for shared objects. 
\end{theorem}
\proof{Proof}
Apply 
\Cref{thm:upperbound}(b)
with $\mathbf{z^*}$ the allocation giving each agent $1/n$ of each object. For tightness, consider $n-1$ objects and $n$ agents who value all objects at $1$. A proportional allocation requires each agent to get exactly $(n-1)/n$ of an object, so all objects must be shared.
\Halmos \endproof

\begin{theorem}
\label{thm:upperbound-ef}
Any instance with $n$ agents admits an envy-free allocation with at most $n-1$ sharings. The $n-1$ is tight even for shared objects. 
\end{theorem}

\proof{Proof}
\citet{Bogomolnaia2017Competitive} prove the existence of a competitive  equilibrium (CE) in the so-called \emph{Fisher market with equal budgets} associated with the division problem, 
when valuations are additive but may have different signs.
Equilibrium allocations satisfy the property of \emph{Pareto-indifference}: if $\mathbf{z}^*$ is an equilibrium allocation and $\mathbf{z}$ gives the same utilities to all agents, then $\mathbf{z}$ is an equilibrium as well.
It is known that competitive equilibrium allocations are envy-free, as well as fractionally Pareto-optimal. This allows us to apply \Cref{thm:upperbound}(b) and get an envy-free allocation with at most $n-1$ sharings.
The proof of tightness is the same as in (a).
\Halmos \endproof

The following theorem was proved in an unpublished manuscript which is no longer available \cite{wilson1998fair}; we repeat the proof briefly for completeness.
\begin{theorem}
\label{thm:upperbound-eq}
Any instance with $n$ agents admits  an equitable allocation with at most $n-1$ sharings. 
The $n-1$ is tight even for shared objects.

\end{theorem}
\proof{Proof}
The following LP finds an equitable allocation which, moreover, maximizes the equal value (such an allocation is often called \emph{max-equitable}):
\begin{align}
\label{eq:equitable}
\text{maximize} \sum_{o=1}^m v_{no} z_{no} &&
\\
\notag
\text{subject to} \sum_{i=1}^n z_{io} = 1 && \text{ for } o\in[m]
\\
\notag
{\sum_{o=1}^m v_{io} z_{io}} = {\sum_{o=1}^m v_{no} z_{no}} && \text{for } i\in[n-1]
\\
\notag
z_{io}\geq 0 && \text{for } i\in[n], o\in [m]
\end{align}
It has $m+n-1$ equational constraints; 
hence an equitable allocation with $n-1$ sharings exists. 
Tightness can be shown using the same argument as in \Cref{thm:upperbound}(b). \Halmos 



\endproof

The following result 
was proved constructively by \citet{goldberg2022consensus}.
Note that it does not follow e.g. from theorems in  \cite{bogomolnaia2016dividing,barman2018proximity,sandomirskiy2022efficient}, since those theorems involve Pareto-improvements of the original allocation, while the  consensus condition is not preserved under Pareto-improvements.

\begin{theorem}
\label{thm:upperbound-cons}
For any positive integers $n,k$,
any instance with $n$ agents admits a $k$-consensus partition with at most $n(k-1)$ sharings. 
The $n(k-1)$ is tight even for shared objects.
\end{theorem}

\proof{Proof}
The following LP finds a $k$-consensus partition:

\begin{align}
\label{eq:consensus}
&& \sum_{i=1}^n z_{io} = 1 && \text{for } o\in[m]
\\
\notag
&& \sum_{o\in[m]} v_{i,o} z_{j,o}  = V_{i}/k && \text{for } i\in[n], j\in[k-1]
\\
\notag
&& 
z_{io}\geq 0 && \text{for } i\in[n], o\in [m]
\end{align}

Note that for $j=k$ the equality holds automatically.
Since there are $m+n(k-1)$ constraints, there exists a basic feasible solution with at most $n(k-1)$ sharings.

For tightness, 
consider an instance with $n$ agents and $n(k-1)$ goods.
For each agent, 
$(k-1)$ of the goods are ``big'' 
and the other $(n-1)(k-1)$ goods are ``small''.
All $n$ sets of big goods are pairwise-disjoint, so that each good is ``big'' for exactly one agent
and ``small'' for all other agents.

Each agent values each of his big goods at $k-0.5$
and each of his small goods at $0.5/(n-1)$.
For each agent, the sum of all values is $k(k-1)$, so in a $k$-consensus partition, the value of each bundle should be exactly $k-1$.
However, the value of each big good to its agent is larger than $k-1$, so all goods must be shared.
\Halmos \endproof

Summarizing the above theorems, we have the following worst-case upper bounds on the number of sharings:
\begin{align*}
S_{PROP}(n) & = S_{EF}(n) = S_{EQ}(n) = n-1
\\
S_{CONS}(n) & = n(n-1)
\end{align*}

\subsection{Computing an allocation attaining the worst-case upper bound}
The previous theorems raise the computational question of how to find an allocation with at most $S_F(n)$ sharings given a fairness criterion F.
The simplex method finds such solutions for any LP, and it may be sufficient for all practical purposes, but its worst-case run-time is exponential.
\citet{khachiyan1980polynomial} showed that the ellipsoid method can be used to find an optimal basic feasible solution in weakly-polynomial time, but its practical performance is not very good. 
Interior-point methods perform better both theoretically and practically, but they may find an interior solution (with many nonzeros) rather than a basic feasible solution.
\citet{megiddo1991finding} gives a strongly-polynomial time algorithm that, given a pair of optimal solutions for the primal and dual problems, finds an optimal basic feasible solution.%
 \ifdefined\FULLVERSION
\footnote{
The above references are based on an answer by Kevin Dalmeijer from Operations Research Stack Exchange (\url{https://or.stackexchange.com/a/3129/2576}).
}
\fi
~\citet{bixby1994recovering} give a different algorithm for the same task, that performs well in practice but does not have good worst-case run-time guarantees.

All the above algorithms still require a solution to the original LP. Here the fairness criteria differ greatly in their computational complexity.

\subsubsection{Proportionality.} it is sufficient to find an optimal solutions to \eqref{eq:m+n-1} and its dual where $U_i = V_i/n$ for all $i\in[n]$; this can be done in weakly-polynomial time. Moreover, a proportional allocation with $n-1$ sharings can be found in strongly-polynomial time by reduction to \emph{connected cake-cutting} in the following way. Arrange the objects on an interval in an arbitrary order, where each object corresponds to a homogeneous sub-interval;
find a connected proportional division of the interval; and allocate the objects accordingly. Since a connected partition makes $n-1$ cuts, at most $n-1$ objects are shared. 
When all valuations are positive or all valuations are negative, a connected proportional allocation of the interval can be found using $O(mn\log(n))$ operations by the Even-Paz 
protocol \citep{Even1984Note}. 

\subsubsection{Envy-freeness.}
Applying the proof of \Cref{thm:upperbound-ef} requires finding a competitive equilibrium (CE) in a Fisher market with equal budgets. With strictly-positive valuations, a CE can be found in $O((n+m)^4\log(n+m))$ operations \citep{orlin2010improved}.
Note that reduction to connected cake-cutting (as we did for proportionality) is not very helpful with envy-freeness.
The reduction yields a piecewise-homogeneous cake; for such cakes, the fastest algorithm that we know of for finding a connected EF allocation requires $O(m^n)$ operations \citep{alijani2017envy} and we do not know whether polynomial-time algorithms exist.

\subsubsection{Equitability.}
Optimal solutions to \eqref{eq:equitable} and its dual can be found in weakly-polynomial time; we do not know if a strongly-polynomial time algorithm exists. 
\begin{open}
	Can an equitable allocation with at most $n-1$ sharings or shared objects be found in strongly-polynomial time?
\end{open}

\subsubsection{Consensus.}
\citet{goldberg2022consensus} present a polynomial-time algorithm that finds a $k$-consensus partition with at most $n(k-1)$ sharings, for any integer $k$.
The main step in their algorithm is Gaussian elimination, which can be done in strongly-polynomial time; hence, their algorithm is in fact strongly-polynomial.

\subsection{Computing an allocation with $s$ smaller than the worst-case}
Given a specific instance, can we decide whether there exists an allocation in which the number of sharings / shared objects is smaller than the upper bound?

In view of \Cref{sub:LP}, a first idea that comes to mind is deciding whether the linear program that corresponds to the required fairness notion admits a solution in which the number of non-zero variables is smaller than the worst-case upper bound. However, this problem is NP-hard in general (see \Cref{sec:lp-bounds}).

Some more specific results are provided by \citet{DBLP:journals/corr/abs-2204-11753}, who study the case of identical valuations. In that case all four fairness notions coincide, and the worst-case upper bound for both sharings and shared objects is $n-1$. 
\citet{DBLP:journals/corr/abs-2204-11753} prove that, when $s$ is the number of \emph{sharings}, deciding the existence of a fair allocation with $s\leq n-2$ is NP-hard, which implies that the computational results for sharings are tight.

In contrast, when $s$ is the number of \emph{shared objects}, finding a fair allocation with for agents with identical additive valuations is polynomial for $s=n-2$ and NP-hard for $s\leq n-3$. So for general additive valuations, the case of $n-2$ shared objects remains open.

\citet{goldberg2022consensus} prove that
it is NP-hard to compute a $2$-consensus allocation with at most $OPT+n-1$ sharings, where $OPT$ is the optimal number of cuts for the instance. 
These results do not imply hardness for $n$-consensus allocation with shared objects. We provide such a hardness proof below.

\begin{theorem}\label{thm:cons-np-hard}
 For any fixed number $n \geq 2$ of agents  with additive valuations, \conswithshared{$n,n-1$} and \conswithsharing{$n,n-1$} are both NP-complete.
\end{theorem}
\proof{Proof}
Membership in {\sf NP} is obvious.
We show a reduction from the $n$-way number partitioning problem, in which given $m$ items of size $x_i$ for $i \in [m]$ summing to $S$, the goal is to decide whether there exist a partition of the items into $n$ bins with equal sum ($S/n$).

Given an instance of the $n$-way number partitioning problem, we construct an instance of \conswithshared{$n,n-1$} with $m+(n-1)$ objects and $n$ agents with the following valuations (where $M$ a sufficiently large integer):

\begin{center}
\begin{tabular}{|c|c|c|c|}
	\hline
	& $m$ number objects & $n-1$ extra objects & Total value \\
	\hline
	Agents $1,\ldots,n-1$:	& $x_i$ & $M$ & $S+(n-1)M$ \\
	\hline
	Agent $n$:           	& $x_i\cdot(1+(n-1)M/S)$ & $0$ & $S+(n-1)M$ \\
	\hline
\end{tabular}
\end{center}

Assume first that there is a solution to the $n$-way number partitioning problem. Take the $n$ bins of the solution and put all the objects of each bin into a different bundle.
Cut the $n-1$ extra objects into $n$ subsets of size $(n-1)/n$ of an object each (this can be done using $n-1$ sharings) and put each subset into a different bundle. 
This yields a solution to \conswithshared{$n,n-1$} and to \conswithsharing{$n,n-1$}.

Assume second that the generated instance has an $n$-consensus partition with at most $n-1$ sharings, or at most $n-1$ shared objects.
In either case, each of the $n-1$ extra objects must be shared otherwise one of these item is entirely allocated to one bundle, and the allocation is not a consensus allocation  for a sufficiently large $M$. This requires $n-1$ sharings, which means that the allocation cannot have any sharings in the $m$ number objects.
Since agent $n$ values the extra objects at $0$, in his eyes only the $m$ number objects matter. Since the allocation is CONS, there must be $n$ bundles in which the $m$ number objects are equally partitioned. This is a solution to the $n$-way number partitioning problem.
\Halmos \endproof

\begin{open}
Given some fixed $n \geq 3$ and $n$ agents with arbitrary additive valuations, what is the run-time complexity of the following problems:
\begin{itemize}
	\item \propwithshared{n, n-2};
	\item \efwithshared{n, n-2};
	\item \eqwithshared{n, n-2};
	\item \conswithshared{n, s} with $s\in [n,n(n-1))$;
	\item \conswithsharing{n, s} with $s\in [n,n(n-1))$?
\end{itemize}
\end{open}
In the following sections we provide partial answers to this question for specific classes of utilities. 

\section{Binary Additive Valuations} \label{sec:binary-utilities}
This section considers agents with additive valuations that are also \emph{binary}, that is, assign a value of either $0$ or $1$ to each object.
\subsection{Unbounded $n$}
\label{sec:binary-unbounded}
We first assume that $n$ is unbounded (i.e. a part of the input), and provide results for envy-freeness.

The results in \citet{DBLP:journals/ai/AzizGMW15, DBLP:conf/aaai/HosseiniSVWX20} imply that 
\efwithsharing{s=0} with binary utilities is NP-complete (by reduction from Exact 3-Cover or Equitable Coloring). We extend their results for any $s\geq 0$.
\begin{theorem}
	\label{thm:unbounded-binary-ef}
    For every fixed integer $s\geq 0$, even when all agents have additive binary valuations, \efwithsharing{s} and \efwithshared{s} are NP-complete.
\end{theorem}
\proof{Proof}

We apply a reduction from \efwithsharing{s=0}. Given an instance of \efwithsharing{s=0}, we construct an instance of \efwithsharing{s} by adding $s$ pairs of new agents, and $s$ additional objects, such that each pair of new agents values a single additional object at 1, and does not value any other object.

Assume first that there exists an allocation for \efwithsharing{s=0}. For each pair of new agents, split the corresponding additional object in half, such that no new agent envies his partner. Given the existing allocation for \efwithsharing{s=0}, no agent envies. Note that there are exactly $s$ sharings, one for each pair of new agents.

Assume second that there exists an allocation for \efwithsharing{s}. The $s$ additional objects must be shared among the pair of new agents since otherwise, one agent will envy the second. So there are at least $s$ sharings, one for each additional object. Therefore, for the other agents, no sharing is allowed, so it is a solution to \efwithsharing{s=0}.

Analogous arguments imply that \efwithshared{s} is NP-complete.
\Halmos \endproof

In contrast, for PROP the problem can be solved in polynomial time for $s=0$.

\begin{theorem}\label{thm:unbounded-binary-prop-no-sharings}
	With binary additive valuations,
    \propwithsharing{s=0} can be solved in polynomial time.
\end{theorem}
\proof{Proof}
We show a reduction to maximum integral network flow:
there is an arc from the source to each agent $i$ with capacity $ \lceil V_i/n \rceil $; 
from each agent $i$ to each object he values at $1$ with capacity $\infty$; 
and from each object to the sink with a capacity $1$.
We prove that there exists a flow with total size $\sum_{i=1}^n \lceil V_i/n \rceil$ if and only if there exists a PROP allocation with no sharings.

Assume first that there is an integral flow
with size $\sum_{i=1}^n \lceil V_i/n \rceil$, which means that all the arcs from the source to the agents are saturated. To each agent $i$, we assign all objects corresponding to arcs with a positive flow. This flow must be $1$ due to the network structure. Note that the agent values all these objects at $1$, otherwise there would be no arc from the agent to the object.
Since the arcs from the source to the agents are saturated, every agent's bundle size is $\lceil V_i/n \rceil \geq V_i/n$, so the allocation is PROP.

Assume second that the instance has a PROP allocation $A$ with no sharing. 
We construct a partial allocation $A'$ by removing, from each agent $i$, all objects he values at $0$, and possibly some objects he values at $1$, until his value becomes exactly $\lceil V_i/n \rceil$. 
Note that $A'$ is a partial allocation as some objects remain unassigned.
We now construct an integral flow as follows. 
For each object $o$, set the flow in the arc $i\to o$  to $1$ iff the object $o$ is assigned to the agent $i$ in allocation $A'$.
For each agent $i$, set the flow from the source to $i$ to the number of objects assigned to $i$ in $A'$, which is exactly $\lceil V_i/n \rceil$ by construction.
For each object $o$, set the flow from $o$ to the sink to $1$.
The flow is valid, since every object is allocated to exactly one agent.
The arcs from the source to the agents are $ \lceil V_i/n \rceil $, hence, they are saturated and the flow size is $\sum_{i=1}^n \lceil V_i/n \rceil$,
\Halmos \endproof

To extend the network-flow algorithm to solve \propwithsharing{s} for $s\geq 1$, we would need to make the capacity of each edge from the source to agent $i$ equal $V_i/n$, and apply 
a variant of the integer network flow which finds the maximum network flow with at most $s$ non-integral edges.
Currently, we do not know of any polynomial-time algorithm for this problem. 
\begin{open}
	(a) Is there a polynomial-time algorithm that, given a flow network and a positive integer $s$, finds a maximum flow in which at most $s$ arcs have a non-integral flow (if such a flow exists)?
		
	(b) What is the running time of 
	\propwithsharing{s} and \propwithshared{s} for any $s\geq 1$, for agents with binary valuations?
\end{open}

For EQ the problem is polynomial too, for any $s\geq 0$. We start from $s=0$.

\begin{theorem}\label{thm:unbounded-binary-eq-no-sharings}
\eqwithsharing{s=0} with binary utilities can be solved in polynomial time.
\end{theorem}
\proof{Proof}
We check, for each $k \in 0,\dots, \lfloor m/n \rfloor$, whether there exists an equitable allocation with common value $k$. This can be done by reduction to maximum-weight matching.
We construct a weighted bipartite agent-object graph in which each agent has $m$ copies. $k$ copies are connected only to objects that the agent values at 1, and their weight is 1; $m-k$ copies are connected only to objects the agent values at 0, and their weight is $\epsilon = 1/(2m)$.\footnote{These $\epsilon$-edges are required since sometimes, there is no EQ allocation where all agents get objects they value at 1, but there are EQ allocations where some agents get objects they value at 0.}. We show that an EQ allocation with common value $k$ corresponds to a matching with weight 
$n k + \epsilon \cdot (m-nk)$.

Assume first that we have an EQ allocation with common value $k$. Then we can match each object $o$ received by an agent $i$
to one of the $m$ copies of agent $i$, depending on $v_{i,o}$: if 
$v_{i,o} = 1$ then the object is matched to one of the $k$ copies along an edge with weight $1$, otherwise it is matched to one of the $m-k$ copies along an edge with weight $\epsilon$.
As each agent $i$ gets $k$ objects he values at $1$, the total weight of the matching is $n k + \epsilon \cdot (m-nk)$.

Assume second that a matching with weight at least
$n k + \epsilon \cdot (m-nk)$ exists.
We construct an allocation by giving each agent $i$ all the objects his copies are matched to. We now show that this allocation is equitable with value $k$.
Since $m\epsilon < 1$,
the matching must contain 
at least $n \cdot k$ edges with weight 
$1$. This means that each of the $k$ copies of each of the $n$ agents must be matched through an edge of weight $1$, so each agent gets $k$ objects he values at $1$.
Moreover, the remaining $(m-nk)$ objects must be matched through an edge of weight $\epsilon$, which means that each of these objects must be allocated to an agent who values it at $0$. 
Hence, we have a complete and equitable allocation with common value $k$.
\Halmos \endproof

To handle the case $s\geq 1$ we use the following lemma.
\begin{lemma}\label{lem:binary-struct}
	If there is an allocation in which all agents have a bundle with an integer value, then there exists an allocation with the same values and with no sharings.
\end{lemma}
\proof{Proof} 
Given an allocation in which all the agents have an integer-valued bundle, we construct an undirected graph where the nodes are the agents, and there is an edge between two agents for each object they share (if two agents share two objects, then there will be two parallel edges between them).
For each edge $i-j$ corresponding to a shared object $o$, we call $i$ a \emph{$0$-endpoint} if $v_{i,o}=0$ and a \emph{$1$-endpoint} $v_{i,o}=1$, and similarly for $j$.
We call an edge a \emph{$00$-edge} if its two endpoints are $0$-endpoints; \emph{$11$-edge} if its two endpoints are $1$-endpoints; and \emph{$10$-edge} otherwise.
We describe a process that iteratively modifies the allocation and the corresponding graph until all edges are removed, which means that the resulting allocation has no sharings.

First, we remove the $00$-edges. For each $00$-edge $i-j$ corresponding to shared object $o$, we move the fraction of $o$ held by $i$ to agent $j$. This removes the edge and does not change the value of any agent.

Next, we remove cycles of $11$-edges. For each cycle of $11$-edges, we compute the smallest fraction of a shared object that is held by one of the agents along the cycle; denote this fraction by $r$. We now move a fraction $r$ of each shared object along the cycle, from the agent who holds it now to the agent following it in the cycle. Since all agents in the cycle value all shared objects at $1$, each agent gives $r$ and receives $r$, so the total value of all agents does not change. By the selection of $r$, at least one edge from the cycle vanishes.

Next, let $i-j$ be any edge remaining in the graph. This edge must have at least one $1$-endpoint, since we removed all $00$-edges; suppose w.l.o.g. that this endpoint is $j$. By the lemma assumption, the total value held by $j$ is an integer, but the edge $i-j$ indicates that $j$ receives a fraction of an object he values at $1$. Therefore, $j$ must receive a fraction of another object, so $j$ must be a $1$-endpoint in another edge, say $j-k$. If either $i$ or $k$ is a $1$-endpoint of the corresponding edge, then by similar arguments, it must be a $1$-endpoint in another edge. As we already removed all cycles of $11$-edges, we can conclude that all remaining edges must be contained in paths that start and end with $0$-endpoints.
For each such path, let $r$ be the smallest fraction of a shared object that is held by one of the agents along the path. Starting with one of the $0$-endpoints, move a fraction $r$ of each shared object from the agent who holds it now to the agent following it in the path. All agents in the path except the endpoints value all shared objects at $1$, whereas the endpoints value the objects they give or receive at $0$. Therefore, the total value of all agents does not change. By the selection of $r$, at least one edge from the path vanishes.

Proceeding this way, we must eventually have a graph with no edges at all, which corresponds to an allocation with no sharings. 
\Halmos \endproof
\begin{theorem}\label{thm:unbounded-binary-eq-sharing}
    For every fixed number $s$, \eqwithsharing{s} with binary utilities can be solved in polynomial time.
\end{theorem}

\proof{Proof} 
We consider two cases, depending on the number of agents.

Case 1: $n > 2 s$.
Then at least one agent does not share any object, which implies that his bundle has an integer value. By EQ, since all the agents' bundle values are equal, all the bundles have an integer value. 
We run \eqwithsharing{s=0} and return the resulting allocation. The running time is polynomial by \Cref{thm:unbounded-binary-eq-no-sharings}, and the outcome is correct by \Cref{lem:binary-struct}.

Case 2: $n \leq 2 s$. Then for every $n \leq 2s$, we can use the polynomial time algorithm for fixed $n$,  described in \Cref{thm:milp-binary} below.
There are at most $2 s$ polynomial-time runs, so overall, the run-time is polynomial.
\Halmos \endproof

The proof of \Cref{thm:unbounded-binary-eq-sharing} is not directly applicable to the shared objects setting, since even for $s=1$, we cannot claim that there is at least one agent with an integer bundle value. To handle this setting, we use a different lemma.

\begin{lemma}\label{lem:eq-shared}
Let $\mathbf{z}$ be an equitable allocation of $[m]$ (the set of all objects). 
Let $S\subseteq [m]$ be the set of objects shared in $\mathbf{z}$.
Then there is a subset $T \subseteq S$ such that there is an equitable allocation of $[m] \setminus T$ with no shared objects at all.
\end{lemma}
\proof{Proof} 
Denote the common value in $\mathbf{z}$ by $V$, and denote $k \coloneq \lfloor V \rfloor$.
Let $\mathbf{z'}$ be the allocation of $[m] \setminus S$ induced by $\mathbf{z}$. 
In $\mathbf{z'}$ no objects are shared, so the value of every agent's bundle is an integer $\leq k$.
If the value of some agent $i$ in $\mathbf{z'}$ is $k-h$, for some integer $h \geq 1$, then we say that $i$ has $h$ ''holes``
(as $i$ is missing $h$ whole objects to reach the value of $k$).
In allocation $\mathbf{z}$, every hole of agent $i$ is ''filled`` by at least two objects of $S$ which are worth 1 to agent $i$.
We map each hole to the objects used for filling it.
For an agent with one hole, this mapping is unique;
for an agent with two or more holes, we select one possible mapping arbitrarily (e.g. if in $\mathbf{z}$ some agent $i$ has $0.5$ of $a$, $0.6$ of $b$ and $0.9$ of $c$, then $i$ has two holes; one option is to map the first hole to $a,b$ and the second hole to $b,c$).

Each hole is mapped to at least two objects; each two holes are mapped to at least three objects; in general, each $t$ holes are mapped to at least $t+1$ objects. Hence, by Hall's theorem, there is a matching between holes and objects, in which each hole is matched to a unique object. Each hole of agent $i$ is matched to a unique object worth 1 to agent $i$. 
By allocating the objects to the matched agents, we get an equitable allocation $\mathbf{z''}$ with common value $k$.
Let $T$ denote the objects from $S$ that are not matched to any hole. Then $\mathbf{z''}$ is an allocation of $[m] \setminus T$ with no shared objects. 
\Halmos \endproof

\begin{theorem}\label{thm:unbounded-binary-eq-shared-object}
    For every fixed number $s$, \eqwithshared{s} with binary utilities can be solved in polynomial time.
\end{theorem}

\proof{Proof} 
For every subset $T$ of at most $s$ objects in increasing order of their cardinality, we look for an equitable allocation of $[m] \setminus T$  with no shared objects. 
By \Cref{thm:unbounded-binary-eq-no-sharings}, we can solve \eqwithsharing{s=0} in polynomial time. Since we run this procedure at most $2^s$ times, the total running time is polynomial in $m$ and $n$ for any fixed $s$.

For every subset $T$, If an equitable no-sharing allocation of $[m] \setminus T$ is found, 
we allocate the shared objects in 
$T$ as follows: for each object in $T$, if an agent values it at 0, we give him the object. If no agent values it at 0, we share the object equally among all the agents. The resulting allocation remains equitable, and we can answer ``yes''.

If no equitable allocation is found for any $T$, then
by \Cref{lem:eq-shared} there is no equitable allocation, so we can answer ``no''.
\Halmos \endproof

We do not yet have results for consensus allocation with an unbounded number of agents.
\begin{open}
What is the running time of \conswithsharing{s} and \conswithshared{s} for any $s\geq 0$, for agents with binary valuations?
\end{open}

\subsection{Fixed $n$}
When the number of agents $n$ is fixed, finding a PROP, EF, EQ or CONS allocation becomes polynomial, 
for every fixed number of agents $n$ and number of sharings/shared objects  $s$.

\begin{theorem}\label{thm:milp-binary}
For every fixed number of agents $n$ and number of sharings $s$:

(a) \propwithsharing{n,s}, \efwithsharing{n,s}, \eqwithsharing{n,s}, \conswithsharing{\allowbreak n,s} with binary utilities can be solved in polynomial time. 

(b) \propwithshared{n,s}, \efwithshared{n,s}, \eqwithshared{n,s}, and \conswithshared{n,s}
can be solved in polynomial time. 
\end{theorem}
\proof{Proof} \footnote{We are grateful to Rohit Vaish for the proof idea for $s=0$.}
If $s\geq n-1$ then the theorem follows from the results in \Cref{sec:bounds}. Therefore we assume $s\leq n-2$.

For subset of agents $N\subseteq [n]$, 
let $M_N$ be the set of objects that are valued at $1$ by all and only the agents in $N$; there are $2^n$ such sets.

We brute-force all possibilities of $\ell\leq s$ objects to share. 
Each shared object can belong to any one of the subsets $M_N$,
and all objects in the same subset are equivalent,
so it is enough to consider at most $(2^n)^{\ell}\leq 2^{ns}$ cases.
We call the $\ell$ shared objects $o_1, o_2, \dots, o_{\ell}$.

In addition, for each shared object $o_k$ for $k\in[\ell]$, we denote by $H_k$ the nonempty subset of agents who get a non-zero fraction from $o_k$. Overall there are $2^n-1$ options for each $H_k$, so at most $(2^n-1)^{\ell}< 2^{ns}$ cases to check.
Overall, for fixed constants $n,s$ we consider only a constant number of cases. 

For part (a), where $s$ is the number of \emph{sharings}, we verify that $\sum_{k \in [\ell]} (|H_k| - 1) \leq s$;
if this does not hold, we discard this case.
For part (b), where $s$ is the number of \emph{shared objects}, this check is not required, it is sufficient that $\ell\leq s$.

For each case, we construct a mixed integer linear program with the following variables:
\begin{itemize}
    \item For every $N \subseteq [n]$ and each agent $i \in [n]$, create an integer variable $x_{i,N}$ representing how many indivisible objects in $M_N$ agent $i$ obtains.
    \item For each $k\in[\ell]$ and each agent $i \in H_k$  create a real variable $y_{i,k}$ representing the fraction agent $i$ gets from shared object $o_k$.
\end{itemize}

Now, we describe the constraints for variables $x_{i,N}, y_{i,k}$ that are true if and only if the corresponding allocation distributes all items among agents: 
\begin{itemize}
    \item For all $N \subseteq [n]$ and $i\in [n]$, 
    $x_{i,N} \geq 0$ and $x_{i,N}$ is an integer number;
    \item For all $N \subseteq [n]$,
    $\sum_{i\in [n]} x_{i,N} = |M_N\setminus \{o_1,\ldots,o_{\ell}\}|$ --- we distribute all non-shared objects from $M_N$ among all agents;
    \item For all $ k  \in [\ell]$ and $i \in H_k$,
    $y_{i,k}\geq 0$  
     --- each agent in $H_k$ gets a non-negative fraction from $o_k$;
    \item For all $ k  \in [\ell]$, $\sum_{i\in H_k} y_{i,k}=1$ --- the divisible object $o_k$ is completely distributed among the agents in $H_k$.
\end{itemize}

Next, we add equations for computing the value each agent $i$ attributes to each bundle $j$
\begin{align*}
u_{i,j} = \sum_{N: N \ni i} 
\bigg( 
x_{j,N} + \sum_{k: ~ H_k \ni j, ~ o_k\in M_N} y_{j,k} 
\bigg)
&&
\forall i,j\in[n].
\end{align*}
The value of agent $i$ is determined only by objects in sets $M_N$ for which $N$ contains $i$. For each such $N$, we add the number of objects given completely to $j$ ($x_{j,N}$), and the fractions of divisible objects $o_k$ given to $j$.
Based on these equations, it is easy to write constraints for any desired fairness notion, according to the definitions in \Cref{sub:fairness-notions}.

An MILP can be solved in polynomial time for any constant number of integer variables \citep{DBLP:journals/mor/Lenstra83}[Section 5].
The number of variables in our MILP-s is bounded by some function of $n$ and $s$, which are fixed constants. So all our MILP-s are solvable in polynomial time. 
\Halmos \endproof

\section{Equal Sum Generalized Binary Valuations}
\label{sec:generalized-binary}
We recall the definition of the generalized binary valuations: for each object $o \in [m]$, there is a rational number $p_o>0$, such that for every agent $i \in [n]$ and object $o \in [m]$,  $v_{i,o} \in \{0, p_o\}$. 
We assume that the sum of the utilities is equal for every agent, so for every two agents, $V_i=V_j$. 

Our main positive result in this section is for the case $n=3, s=1$.

\begin{theorem}\label{thm:g-binary}
\propwithshared{3,1} with equal-sum generalized binary valuations can be solved in polynomial time.
\end{theorem}

To design our algorithm, we use another algorithm for the max-min variant of the $n$-way partition problem in which $s=1$ object is allowed to be shared. The problem is defined as follows: given 
$m$ objects 
and
$n=3$ agents with identical valuations,
return an allocation 
with at most $s=1$ shared object, 
for which the smallest bundle value is as large as possible. We call this problem \maxminnRwayidenticalsplitpar{3, 1}.
A polynomial-time algorithm for \maxminnRwayidenticalsplitpar{3, 1} is given in  \citep{DBLP:journals/corr/abs-2204-11753}. Their proof uses two structure Lemmas, that we will use later:
\begin{lemma}\label{lem:maxminstructure}
    In any instance of \maxminnRwayidenticalsplitpar{3, 1} either the output is perfect (all the bundle values are equal), or, the shared object is shared only between the two smallest bundles, say bundle 1 and bundle 2, and their values are equal.
\end{lemma}
\begin{lemma}\label{lem:largestItems}
With identical valuations, 
for every allocation with $s = 1$ shared object and bundle sums $b_1,b_2,b_3$, there exists an allocation with the same bundle sums $b_1,b_2,b_3$ in which only the \emph{highest-valued} object is shared. 
\end{lemma}

When there are three agents, it is easy to visualize their valuations using a table that we introduce next. Consider \cref{tab:g-binary}, that divide the set of objects into seven different categories, from ${\cal X}_1$ to ${\cal X}_7$.
\begin{table}[H]
    \centering
    \begin{tabular}{c || c || c || c }
       ${\cal X}_l$ & $\sum_{j \in {\cal X}_l} u_1(x_j)$ & $\sum_{j \in {\cal X}_l} u_2(x_j)$ & $\sum_{j \in {\cal X}_l} u_3(x_j)$ \\
       \hline
       \hline
       ${\cal X}_1$  & $\sumgb{1}=\sum_{j \in {\cal X}_1} p_j$ & $\sumgb{1}=\sum_{j \in {\cal X}_1} p_j$ & $\sumgb{1}=\sum_{j \in {\cal X}_1} p_j$\\
       \hline
       ${\cal X}_2$  & $\sumgb{2}=\sum_{j \in {\cal X}_2} p_j$ & $\sumgb{2}=\sum_{j \in {\cal X}_2} p_j$ & 0 \\
       \hline
       ${\cal X}_3$  & $\sumgb{3}=\sum_{j \in {\cal X}_3} p_j$ & 0 & $\sumgb{3}=\sum_{j \in {\cal X}_3} p_j$\\
       \hline
       ${\cal X}_4$  & 0 & $\sumgb{4}=\sum_{j \in {\cal X}_4} p_j$ & $\sumgb{4}=\sum_{j \in {\cal X}_4} p_j$\\
       \hline
       ${\cal X}_5$  & $\sumgb{5}=\sum_{j \in {\cal X}_5} p_j*$ & 0 & 0\\
       \hline
       ${\cal X}_6$  & 0 & $\sumgb{6}=\sum_{j \in {\cal X}_6} p_j*$ & 0\\
       \hline
       ${\cal X}_7$  & 0 & 0 & $\sumgb{7}=\sum_{j \in {\cal X}_7} p_j*$\\
    \end{tabular}
    \caption{equal-sum generalized binary utilities for three agents}
    \label{tab:g-binary}
\end{table}
There is no interest in assigning an object to an agent that values it at zero, so we automatically assign objects in ${\cal X}_5$ to agent 1, objects in ${\cal X}_6$ to agent 2, and objects in ${\cal X}_7$ to agent 3. To notify that the entire set of objects is assigned to an agent, we add a star to the corresponding cell.
Denote by $V$ the sum of the utilities. Note that:
$$
V=\sumgb{1}+\sumgb{2}+\sumgb{3}+\sumgb{5}=\sumgb{1}+\sumgb{2}+\sumgb{4}+\sumgb{6}=\sumgb{1}+\sumgb{3}+\sumgb{4}+\sumgb{7},$$ 
which implies
$$
\sumgb{3}+\sumgb{5}=\sumgb{4}+\sumgb{6} ; \text{\hspace{0.3in}} \sumgb{2}+\sumgb{6}=\sumgb{3}+\sumgb{7}; \text{\hspace{0.3in}} \sumgb{2}+\sumgb{5}=\sumgb{4}+\sumgb{7}.
$$
Assume w.l.o.g. that $\sumgb{5} \geq \sumgb{6} \geq \sumgb{7}$, so $\sumgb{2} \leq \sumgb{3} \leq \sumgb{4}$. 


Using this notation, we present  the algorithm 
for proving \Cref{thm:g-binary}.
The algorithm is based on a detailed case analysis, which we provide in  \Cref{app:g-binary}.
In general, there are three main cases:
\begin{itemize}
    \item $\sumgb{1} \leq \sumgb{4}$ --- there is always a PROP allocation. The intuitive meaning of this case is that some set of objects (namely ${\cal X}_4$) is worth a lot to agents 2 and 3, but worth little to agent 1. We can allocate these objects among agents 2 and 3 in any way we like with a single shared object, and use the other objects (which are less valuable) to compensate agent 1 without additional splits (see \ref{sec:1-smaller-4} in \Cref{app:g-binary});
    \item $\sumgb{4} < \sumgb{1} \leq \frac{2}{3} \cdot V$ --- there is always a PROP allocation. 
    The intuitive meaning of this case is that there is a set of objects objects (namely ${\cal X}_1$) that are worth a lot to all three agents, but the sum of these objects is at most $2/3$ of the total value $V$. Therefore, it is enough to allocate this set among two of the three agents (giving each of them at least $V/3$), and use the other objects to compensate the third agent (see \ref{sec:4-smaller-1} in \Cref{app:g-binary});
    \item $\sumgb{1} > \frac{2}{3} \cdot V$.
The intuitive meaning of this case is that there is a set of objects objects (namely ${\cal X}_1$) that are worth a lot to all three agents, and their sum is larger than $2/3$ the total value. Therefore, we must partition these objects among all three agents. To do this, we use \maxminnRwayidenticalsplitpar{3,1}.    
If 
the largest bundle sum is at most
$\frac{\sumgb{1}+\sumgb{2}+2 \cdot (\sumgb{3}+\sumgb{4}+\sumgb{5}+\sumgb{6})}{3}$, then a PROP allocation exists and we answer ``yes'';
otherwise, we prove that no PROP allocation can exist, so we answer ``no'' (see \ref{sec:1-big} in \Cref{app:g-binary}).
\end{itemize}

This is a polynomial time algorithm for \propwithshared{3,1} with equal-sum generalized binary utilities, proving \Cref{thm:g-binary}.

\begin{open}
Given some fixed $n\geq 3$,
what is the run-time complexity of the problems
\propwithshared{n,n-2}, 
\efwithshared{n,n-2}, \eqwithshared{n,n-2}, and \conswithshared{n,n-2} for generalized binary valuations?
\end{open}

\section{Non-Degenerate Valuations}
\label{sec:non-degenerate}

As mentioned in the introduction, 
\citet{sandomirskiy2022efficient}
prove that problems \propfpowithsharing{n,s}, \effpowithsharing{n,s}, \propfpowithshared{\allowbreak n,s}, \effpowithshared{n,s}
are solvable in time $O(poly(m))$ for any fixed $n$ and $s$, whenever the valuations are \emph{non-degenerate}, and argue --- somewhat informally --- that ``almost all'' valuations are non-degenerate; this means that almost all instances of the above problems are easy.
In this section we would like to show that the requirement of fractional Pareto-optimality (fPO) is essential for this result: when this requirement is dropped, or even just relaxed to dPO, computational hardness strikes even for non-degenerate valuations, and it is no longer true that almost all instances are easy.

To prove this statement, we first have to formally define the notion of ``almost all valuations are easy''.

\subsection{Definitions}

\newcommand{\inputs}{\mathcal{I}({t,c})}
\newcommand{\inputsa}{\mathcal{I}({t_1,c_1})}
\newcommand{\inputsb}{\mathcal{I}({t_2,c_2})}
\newcommand{\inputsba}{\mathcal{I}({t_2(t_1),c_2})}
\newcommand{\pinputs}{\mathcal{I}_p({t,c})}
\newcommand{\pinputsa}{\mathcal{I}_p({t_1,c_1})}
\newcommand{\pinputsb}{\mathcal{I}_p({t_2,c_2})}
\newcommand{\pinputsba}{\mathcal{I}_p({t_2(t_1),c_2})}
\newcommand{\ball}{\mathcal{B}}

We consider a decision problem $P$, whose input is a vector of some $t$ non-negative integers.
As the input size is measured by $t$, the binary encoding length of the numbers should be polynomial in $t$. For simplicity, we assume that 
all input integers are in the range $[0, 2^{c t}]$, for some constant $c$ that may depend on the problem. 
We denote the set of possible inputs of size $t$ (that is, $t$-sized vectors of integers in $[0, 2^{c t}]$) by $\inputs$. 

For any input vector $x\in \inputs$, 
we denote its infinity norm by $\|x\|$. For any $x$ and integer $r$, we denote by $\ball(x,r)$ the \emph{ball of radius $r$} around $x$: 
$\ball(x,r) := \{x'\in \inputs: \|x-x'\|\leq r\}$.

The notion of ``almost all instances are easy'' is formalized by the following definition of ``generically polynomial-time algorithm''.



\begin{definition}
\label{def:generically-easy}
Given an algorithm $A$ for a problem $P$, 
we say that \emph{$A$ runs generically in polynomial time}
if there exists a polynomial function $f_p$ such that, for every size $t$ and input $x\in \inputs$,
there is a subset of ``Good inputs'' $G(x)\subseteq \ball(x,f_p(t))$, such that the following holds:

(a) Algorithm $A$ runs in time poly($t$) on all inputs in $G(x)$;

(b) The fraction of good inputs approaches 1, that is,
\begin{align*}
\lim_{t\to\infty}
\min_{x\in \inputs}\frac{|G(x)|}
{|\ball(x,f_p(t))|} = 1,
\end{align*}

(c) Given $x$, it is possible to compute in time poly($t$) a vector  in $G(x)$.
\end{definition}

The results in \citet{sandomirskiy2022efficient} imply:
\begin{proposition}
\label{prop:generically-polynomial}
For every fixed $n,s$, decision problems \propfpowithsharing{n,s}, \effpowithsharing{n,s}, \propfpowithshared{n,s}, \effpowithshared{n,s} have algorithms that run in generically-polynomial time.
\end{proposition}

\proof{Proof}
In these problems, the input is $x \equiv v = $ a valuation matrix, and the number of integers in the input is $t = m n = $ the number of values in the valuation matrix.
We choose the polynomial $f_p(t) := t^3$.
We define the set $G(v)$ as the subset of matrices in $\ball(v,t^3)$ that define non-degenerate valutions. We show that this set satisfies \Cref{def:generically-easy}.

(a) Is satisfied by the algorithms in \cite{sandomirskiy2022efficient}.

(b) For each input (valuation matrix) $v'\in \ball(v, t^3)$, for each value $v'_{i,o}$ in the matrix
there are some $r_{i,o}$ options, where $t^3+1 \leq r_{i,o} \leq 2 t^3 + 1$.%
\footnote{
$r_{i,o} = t^3+1$ when the original value $v_{i,o}$ is at the bottom or top end of the allowed interval, and
$r_{i,o} = 2 t^3 + 1$ when $v_{i,o}$ is at the middle of the allowed interval.
}
The inputs in $G(v)$ are the matrices $v'$ in which each value $v'_{i,o}$ (the value of agent $i$ to object $o$) satisfies inequalities of the form:
$v'_{i,o} / v'_{j,o} \neq v'_{i,p} / v'_{j,p}$ for other agents $j$ and objects $p$; the number of such inequalities is smaller than $m n = t$. 
Therefore, the number of options to choose $v'_{i,o}$ is at least $r_{i,o} - t$.
Overall,
\begin{align*}
\frac{|G(v)|}{|\ball(v, t^3)|} 
&\geq 
\frac
{\prod_{i,o} (r_{i,o}-t)}
{\prod_{i,o} (r_{i,o})}
\\
&
=
{\prod_{i,o} (1-t/r_{i,o})}
\geq
(1 - t/(t^3+1))^t,
\end{align*}
which approaches $1$ as $t\to\infty$;
this formalizes the claim that ``almost all inputs are non-degenerate''.

(c) To compute an input in $G(v)$, it is sufficient to check at most $t$ options for changing each coefficient $v_{i,o}$; this can be done using polynomial in $t$ time.
\Halmos \endproof

Our goal is to prove that some problems do \emph{not} have generically-polynomial-time algorithms. To this end, we first define a \emph{multi-reduction}. We present a simplified version first, and then the full version.

\begin{definition}[multi-reduction -- simplified]
\label{def:multireduction-simplified}
Given two decision problems $P_1$ and $P_2$, 
both defined on inputs in $\inputs$ for some $c\geq 1$,
a \emph{polynomial-time multi-reduction} from $P_1$ to $P_2$ 
is a family of functions, $h_t: \inputs \to \inputs$, 
which maps an input for $P_1$ to an input for $P_2$, and satisfies the following:

(a) $h_t$ runs in time poly$(t)$;

(b) There exists a super-polynomial function $f_e$ such that, for all $t$ and all $x_1\in \inputs$, when $x_2 := h_t(x_1)$, 
\begin{align*}
P_2(x_2') = P_1(x_1)
&&
\text{ for all }
x_2' \in \ball(x_2, f_e(t) ),
\end{align*}
\end{definition}
A multi-reduction is stronger than a usual reduction in that each input to $P_1$ is transformed simultaneously to an super-polynomially-large set of inputs to $P_2$, all of which have the same output.

\Cref{def:multireduction-simplified} is ``simplified'' since it assumes that the  size and encoding length of the inputs to $P_1$ is equal to the size and encoding length of the inputs to $P_2$. 
In fact, reductions often add some inputs or increase the encoding length.
We handle this technical issue by assuming that the input $x_1$ is in $\inputsa$ and the transformed input $x_2$ is in $\inputsb$, where $t_2$ depends polynomially on $t_1$ but they do not have to be equal.
We also allow the constants $c_1$ and $c_2$ to differ.
\begin{definition}[multi-reduction -- full]
\label{def:multireduction-full}
Given two decision problems, $P_1$ defined on inputs in $\inputsa$ for some constant $c_1\geq 1$ and $P_2$ defined on inputs in $\inputsb$ for some constant $c_2\geq 1$, 
a \emph{multi-reduction} from $P_1$ to $P_2$  consists of a polynomial-time-computable function
$t_2: \mathbb{N} \to \mathbb{N}$ and
a family of functions, $h_{t_1}: \inputsa \to \inputsba$, 
which maps an input for $P_1$ to an input for $P_2$, and satisfies the following:

(a) $h_{t_1}$ runs in time poly$(t_1)$;

(b) There exists a super-polynomial function $f_e$ such that, for all $t_1$ and all $x_1\in \inputsa$, when $x_2 := h_{t_1}(x_1)$, 
\begin{align*}
P_2(x_2') = P_1(x_1)
&&
\text{ for all }
x_2' \in \ball(x_2, f_e(t_2) ),
\end{align*}
\end{definition}

\begin{definition}
\label{def:generically-np-hard}
A decision problem $P_2$ is called \emph{generically NP hard} if there exists a multi-reduction from some NP-hard problem $P_1$ to $P_2$.
\end{definition}

We prove that the relation between ``generically-polynomial''  and ``generically-NP-hard'' is analogous to the relation between ``polynomial'' and ``NP-hard'':
\begin{proposition}
\label{thm:generically-easy-is-not-generically-hard}
If a problem 
is generically-NP-hard, then it does not have  a generically-polynomial-time algorithm unless \textsc{P=NP}.
\end{proposition}

\proof{Proof}
Suppose by contradiction that some decision problem $P_2$ simultaneously satisfies the following:

(a) There is a generically-polynomial-time algorithm $A_2$ for $P_2$; denote by $f_p$ the polynomial function in \Cref{def:generically-easy}.

(b) There is a multi-reduction $h_{t_1}$ from some NP-hard problem $P_1$ to  $P_2$; denote by $f_e$ the super-polynomial function in \Cref{def:generically-np-hard}.

Since $f_p$ is polynomial and $f_e$ is super-polynomial, there is some $t_0$ such that $f_e(t) > f_p(t)$ for all $t>t_0$.
We show a polynomial-time algorithm for solving $P_1$ on all inputs of size $t_1$ such that $t_2(t_1)>t_0$.

Given an input $x_1\in \inputsa$ to problem $P_1$, we
use the assumed multi-reduction $h_{t_1}$ to compute in polynomial time an input $x_2 \in \inputsb$ to problem $P_2$.

By \Cref{def:generically-easy}(a,c), it is possible to compute in time poly$(t_2)$
another input $x_2' \in \ball(x_2, f_p(t_2))$, such that $A_2$ runs on $x_2'$ in time poly$(t_2)$.

But by definition of multi-reduction (\Cref{def:multireduction-full}), 
$P_2(x_2'') = P_1(x_1)$ for all $x_2'' \in \ball(x_2, f_e(t_2))$.
In particular, since $f_e(t_2) > f_p(t_2)$, this holds for $x_2'$ too, so 
$P_2(x_2') = P_1(x_1)$.

Therefore, by running $A_2$ on $x_2'$, we get the correct answer to $P_1(x_1)$.
\Halmos \endproof

Now we are ready for our main results: showing that fair division problems without the fPO requirement are generically-NP-hard, and therefore probably do not have generically-polynomial-time algorithms (in particular, they are NP-hard even for non-degenerate valuations).

\subsection{EF and PROP allocations --- fixed $n$}

\begin{theorem}
\label{thm:hardness-fair-with-sharings}
The decision problems \propwithsharing{n,s} and \efwithsharing{n,s} 
are generically NP-hard
in the following cases:

(a) For any fixed $n \geq 2$ and $s=0$;

(b) For any fixed $n\geq 3$ and $s\range{0}{n-3}$.

\end{theorem}

\proof{Proof}
We show a multi-reduction from $P_1 = $ \textsc{$k$-way Partition}
to $P_2 = $ \fairwithsharing{n,s}.

The input to $P_1$ 
is a list $D := [d_1,\ldots,d_{t_1}]$ of integers\footnote{
In this section we use $D$  for the input to \textsc{$k$-way Partition} (instead of $x$), because $x$ is used for a generic input vector in $\inputs$.} in $[0,2^{c_1 t_1}]$, where $c_1\geq 1$ is some integer constant%
\footnote{
In fact, we can assume $c_1=1$.
The NP-hardness proof of \textsc{Partition} described in 
\citet{garey1979computers}
constructs instances with some $p$ integers, in which each input integer can be represented by at most $3 p \log(p)$ bits. We can add to each constructed instance $3 p \log(p) - p$ zeros; this does not affect the existence of a partition.
Then, taking $t_1 := 3 p \log(p)$, we have that each constructed instance has $t_1$ integers in $[0, 2^{t_1}]$.
So \textsc{Partition} is NP-hard even when restricted to $\inputs$ with $c=1$.
}.
The task is to decide if the integers can be partitioned into $k$ subsets with sum $S := \frac{1}{k} \sum_{o=1}^{t_1} d_o$.
We denote $L := 2^{c_1 t_1}$.

\textbf{For (a),} we reduce from \textsc{$k$-way Partition} for $k=n$.
We construct an instance of \fairwithsharing{n,0} with $m = t_1+n$ objects. Note that the number $t_2$ of inputs to $P_2$ is the number of values in the valuation matrix, $t_2 = m n = n t_1 + n^2$, which is polynomial in $t_1$ as $n$ is fixed.

All agents' valuations will be integers in the range 
$[0, 2^{c_2 t_2}]$, where $c_2 := 2 c_1$. 
For simplicity, we denote by ``$/$'' the integer division operation.

\begin{itemize}
\item The objects $o \in \{1,\ldots, t_1\}$ are \emph{usual objects}. 
Each agent $i\in[n]$ values each usual object $o$ at $L \cdot d_o + L/(16 m)$.
\item The objects $o \in \{t_1+1,\ldots, t_1+n\}$ are \emph{compensation objects}. 
Each agent $i\in[n]$ values the compensation object $t_1+i$ at 
$3 L / 8$,
and every other compensation object at 
$L / (16 m n)$.
\end{itemize}
Note that all valuations are integers smaller than $L^2 + L$, which is indeed smaller than $2^{c_2 t_2}$.

For the multi-reduction we will use the function
\begin{align*}
f_e(t_2) &:= L/(64 m n) = 2^{t_1 c_1 - 6} / (m n) 
\\
&
= 2^{c_1 (t_2 - n^2)/n} / t_2.
\end{align*}
Note that $f_e$ is indeed a super-polynomial function of $t_2$.

Suppose that there exists a partition of $D$ into $n$ subsets $(D_1,\ldots, D_n)$ with sums equal to $S$.
We construct an allocation of the objects by giving each agent $i$ the set of usual objects corresponding to the integers in $D_i$, as well as the compensation object indexed $t_1+i$.
We show that the resulting allocation is EF (and hence PROP) not only for the constructed valuation matrix, but also for any valuation matrix with infinity-distance of at most $L/(64 m n)$.

Indeed, in every such input, each agent values usual object $o$ at least $L\cdot d_o + L/(16 m) - L/(64 m n) > L\cdot d_o$,
and values his compensation object at least $3 L / 8 - L/(64 m n) > 2 L / 8$, 
so he values his own bundle at least $L S+ 2 L /8$.

On the other hand, each agent values usual object $o$ at most $L\cdot d_o + L/(16 m) + L/(64 m n) < L\cdot d_o + L/(8 m)$, and values any other compensation object at most $L/(16 m n) + L/(64 m n) < L/(8 m n)$.
Therefore, each agent values the bundle of any other agent at most $L S + m\cdot (L / 8 m) = L S + L/8$, so there is no envy.

Conversely, suppose there is a PROP allocation (or an EF allocation, which is always PROP). 
As each agent values each usual object $o$ at least $L d_o$, the sum of \emph{all} values for the agent is larger than $L\cdot n S$,
so PROP requires each agent to receive a value larger than $L S$.

The value of each potential bundle comes from integer multiples of $L$ (--- the $L\cdot d_o$ value of usual objects), and fractions of $L$ (--- the $L/(16 m)$ value of usual objects, and the values of compensation objects). The sum of all fractions is strictly less than $L$.
Therefore, to get a value of at least $L S$, each agent must receive usual objects with $\sum_o d_o \geq S$. The partition of usual objects corresponds to a partition of $D$ into $n$ subsets of sum at least $S$ each. But the total sum of all integers is $n S$, so the sums of all $n$ subsets must be exactly $S$.

\textbf{For (b),} we reduce from $k$-way Partition for $k=n-s-1$. Note that $k\geq 2$ since $s\leq n-3$.
We construct an instance of \fairwithsharing{n,s}
with $n$ agents, of whom $k$ are \emph{usual agents} constructed similarly to the $k$ agents in part (a), and the remaining $s+1$ are \emph{special agents}, indexed $k+1,\ldots, k+s+1$.
As in part (a), there are $t_1$ usual objects and $n$ compensation objects. In addition, there is a single \emph{special} object. 
All in all we have $n = k+s+1$ and $m=t_1+n+1$. The additional valuations are:
\begin{itemize}
\item Each usual agent values the special object at $(s+1) L S + L/16$
\item Each special agent $i$ values the special object at 
$n\cdot L S + L/16$,
and values each of the other objects (usual and compensation) at $L/(16mn)$.
\end{itemize}
We use the same super-polynomial function $f_e(t_2) = L/(64 m n)$ as in part (a).

Suppose there exists a partition of $D$ into $k$ subsets $(D_1,\ldots, D_k)$ with sums equal to $S$.
Similarly to part (a), we give each usual agent $i$ the set of usual objects corresponding to the integers in $D_i$, as well as the compensation object indexed $t_1+i$.
We divide the special object equally among the $s+1$ special agents, so that there are $s$ sharings.
We show that the resulting allocation is EF (and hence PROP) not only for the constructed valuation matrix, but also for any valuation matrix with infinity-distance of at most $L/(64 m n)$.

As in part (a), each usual agent $i\in\{1,\ldots, k\}$ values his own bundle at least $L S + 2 L / 8$, 
and values the bundle of every other usual agent at most $L S + L/8$.
Moreover, each usual agent values the special object at most $LS + L/16 + L/(64 m n) < LS + L/8$, so values the bundle of every special agent at most $LS + (L/8)/(s+1)$, so the usual agents do not envy.

Each special agent $i\in\{k+1,\ldots, k+s+1\}$ values the special object at least $nLS +L/16 - L/(64 m n) > nLS$, so values his own bundle at least $(n L S)/(s+1) > L S$.
Moreover, each special agent values every other object at most $L/(16 m n) + L/(64 m n) < L/(8 m n)$, so values every other bundle at most $L/8$. Therefore, the special agents do not envy too, so the allocation is EF and hence also PROP.

Conversely, suppose there is a PROP allocation (or an EF allocation, which is PROP) with at most $s$ sharings. 
Then:

- Each special agent must receive a value of at least $(n L S)/ n = L S$. Since their total value for all objects except the special object is at most $L/8$, each special agent must receive a part of the special object. This means that all $s$ sharings are in the special object, so the usual and compensation objects must be allocated without sharing.

- For each usual agent, the sum of all values is larger than $k L S + (s+1) L S = n L S$, so they must receive a value larger than $L S$ each.
As in part (a), the sum of all fractions in the values of usual and compensation objects is smaller than $L$. 
 Furthermore, no more sharings are allowed in the special object. So each usual agent must receive usual objects with $\sum_o L d_o \geq L S$, and this induces a $k$-partition of $D$ into $k$ subsets of sum exactly $S$.
\Halmos \endproof

We could not adapt the proof of \Cref{thm:hardness-fair-with-sharings} to \emph{shared objects}. We could potentially have $s$ special objects that have to be shared, but then some parts of the special objects could benefit the usual agents, which would undermine the claim in the last paragraph of that proof. 
We can still prove similar results for shared objects, but with a different multi-reduction.

\begin{theorem}
\label{thm:hardness-fair-with-shared}
The problems \propwithshared{n,s} and \efwithshared{n,s} 
are 
generically NP-hard
in the following cases:

(a) For any fixed $n \geq 2$ and $s=0$;

(b) For any fixed $n\geq 3$ and $s\range{0}{n-3}$.
\end{theorem}

\proof{Proof}
We prove both parts using a single multi-reduction from
the case of identical valuations, proved to be NP-hard in \cite{DBLP:journals/corr/abs-2204-11753}.

We are given a list $D$ of items%
\footnote{We use ``items'' for the original instance $D$, and ``objects'' for the constructed fair allocation instance.}
whose values are integers $[d_1,\ldots,d_{t_1}]$
in $[0, 2^{c_1 t_1}]$, where $c_1\geq 1$ is some constant,
with $\sum_{o=1}^{t_1} d_o = n S$.
We have to decide if they can be partitioned into $n$ bins with sum $S$ each, with at most $s$ items split between two or more bins.
Let $L := 2^{c_1 t_1}$.

We construct a fair allocation instance with $n$ agents and $m = t_1+n$ objects. Hence, $t_2 = m n = n t_1 + n^2$, which is linear in $t_1$ as $n$ is fixed.
All agents' valuations will be integers in $[0, 2^{c_2 t_2}]$, where $c_2 := 2 c_1$.
Recall that ``$/$'' denotes integer division.
\begin{itemize}
\item The objects $o \in \{1,\ldots, t_1\}$ are \emph{usual objects}. 
Each agent $i\in[n]$ values each usual object $o$ at $L \cdot d_o + L / (16 m n)$.
\item The objects $o \in \{t_1+1,\ldots, t_1+n\}$ are \emph{compensation objects}. 
Each agent $i\in[n]$ values the compensation object $t_1+i$ at 
$ 3 L / (8n)$
and every other compensation object at 
$L / (16 m n)$.
\end{itemize}
Note that all valuations are integers smaller than $L^2 + L$, which is indeed smaller than $2^{c_2 t_2}$.
For the multi-reduction we use the same  super-polynomial function as in \Cref{thm:hardness-fair-with-sharings}: 
$f_e(t_2) := L/(64 m n)$.

Suppose that there exists a partition of $D$ into $n$ subsets $(D_1,\ldots, D_n)$ with sums equal to $S$, with at most $s$ shared items.
We construct an allocation of the objects in our instance by giving each agent $i$ the set of usual objects  corresponding to the items in $D_i$ (including fractions), as well as the compensation object $t_1+i$.
Note that the number of shared items remains unchanged (at most $s$).

We now prove that the allocation is EF (and hence PROP) not only for the constructed instance, but also for any instance with infinity distance at most $L/(64 m n)$.
In all such instances, each agent values every usual object $o$ at least $L d_o + L/(16 m n) - L/(64 m n) > L d_o$  and his compensation object at least $3 L / (8 n) - L/(64 m n) > 2 L / (8 n)$, so values his own bundle at least $L S + 2 L/(8 n)$; and values the bundle of every other agent at most $S+L/(8n)$. Hence the allocation is EF and also PROP.

Conversely, suppose there is a PROP allocation (or an EF allocation, which is always PROP), with at most $s$ shared objects.
For each agent $i$, the sum of all object values is larger than $L \cdot n S$,
so PROP requires each agent to receive a value larger than $L S$. 

Denote by $D_i$ the set of items from $D$ (including fractions) corresponding to the usual objects given to agent $i$.

The value of each potential bundle comes from integer multiples of $L$ (--- the $L\cdot d_o$ value of usual objects), and fractions of $L$ (--- the $L/(16 m n)$ value of usual objects, and the values of compensation objects). The sum of all fractions is strictly less than $L/n$.
Therefore, to get a value of at least $L S$, each agent must receive usual objects with 
$\sum_{o\in D_i} L d_o > L S-L/n$,
which implies 
$\sum_{o\in D_i} d_o > S-1/n$ for each agent $i\in[n]$.
Since 
$\sum_{o\in D} = n S$,
this implies 
$\sum_{o\in D_j} d_o < S+1$ for each agent $j\in[n]$.
Therefore, each bundle $D_i$ with 
$\sum_{o\in D_i} d_o \neq S$
must have fractions of items (as all item values in $D$ are integers).
Therefore, we can move fractions of items from bundles $D_i$ with sum larger than $S$ to bundles $D_j$ with sum smaller than $S$, without increasing the number of shared items. 
By iteratively moving these fractions, we can construct in polynomial time a partition of $D$ with at most $s$ shared items, in which all bin sums equal $S$.
\Halmos \endproof

We could not adapt the proof of \Cref{thm:hardness-fair-with-shared} to sharings, since moving fractions of items between bins might increase the number of sharings. This is why we used a different proof for sharings in \Cref{thm:hardness-fair-with-sharings}.

The hardness of \Cref{thm:hardness-fair-with-sharings}(a)  and
\Cref{thm:hardness-fair-with-shared}(a) 
remains even if we add the requirement of \emph{discrete} PO (in contrast to fractional PO). We consider only part (a) as the requirement of discrete PO makes sense only for allocations without sharing.

\begin{theorem}
\label{thm:hardness-dpo-with-sharings}
For any fixed integer $n\geq 2$, 
the decision problems
\propdpowithsharing{n,0} ($\equiv$ \propdpowithshared{n,0})
and
\efdpowithsharing{n,0} ($\equiv$ \efdpowithshared{n,0})
are 
generically NP-hard.
\end{theorem}

\proof{Proof}
We adapt the proof 
of \Cref{thm:hardness-fair-with-sharings}(a).
It is sufficient to prove that, 
if there exists an equal-sum partition of $D$, then there exists a dPO+EF allocation (with no sharings) of the objects. 

Let $\mathbf{z}$ be the EF allocation constructed in the proof 
of \Cref{thm:hardness-fair-with-sharings}(a).
If $\mathbf{z}$ is also dPO then we are done. Otherwise, let $\mathbf{y}$ be a dPO allocation (with no sharing) that Pareto-dominates 
$\mathbf{z}$.
We claim that $\mathbf{y}$ is EF too.

Since $\mathbf{y}$ dominates $\mathbf{z}$, 
each agent must receive in $\mathbf{y}$ a value of at least $L S+2L/8 > LS$.
As in the proofs of the previous theorems, 
this means that each agent must receive some usual objects corresponding to integers from $D$ with a sum at least $S$. But since the sum of all integers in $D$ is $n S$, each agent must receive usual objects  corresponding to integers of sum exactly $S$. 
This means that each agent $i$ values the bundle of every other agent at most $L S+L/8$. 
Therefore, the allocation $\mathbf{y}$ is EF, so it is the desired dPO+EF allocation.
\Halmos \endproof

\Cref{thm:hardness-dpo-with-sharings} is interesting as it shows a crucial difference between the apparently-similar concepts fPO and dPO: whilst the fPO allocations can be enumerated in polynomial time (as their number is polynomial in $m$ when the valuations are non-degenerate \cite{sandomirskiy2022efficient}), the dPO allocations cannot.

When $s=n-1$, an EF and PROP allocation with $s$ sharings always exists (see \Cref{sec:bounds}). Therefore, only the case $s=n-2$ remains open.
With identical valuations, this case is NP-hard with sharings and polynomial with shared objects  \cite{DBLP:journals/corr/abs-2204-11753}. We do not know if the same is true with non-degenerate valuations.
\begin{open}
For any $n\geq 3$ and $s=n-2$, 
do the problems
\propwithsharing{n,s},
\efwithsharing{n,s},
\propwithshared{n,s},
\efwithshared{n,s}
have a generically-polynomial-time algorithm?
\end{open}

For equitability and consensus allocation, we currently have neither a generically-polynomial-time algorithm nor a proof of generic NP-hardness.

\begin{open}
(a) For any $n\geq 2$ and $s \in [0,n-1)$, 
do the problems
\eqwithsharing{n,s}, \eqwithshared{n,s} 
have a generically-polynomial-time algorithm?

(b) For any $n\geq 2$ and $s \in [0, n(n-1))$, 
do the problems
\conswithsharing{n,s}, \conswithshared{n,s}
have a generically-polynomial-time algorithm?
\end{open}

\subsection{EF and PROP allocations --- unbounded $n$}
So far we assumed that $n$ is fixed. If $n$ is unbounded (part of the input), then we can prove \emph{strong} NP-hardness for both sharings and shared objects.
This requires to adapt the notions of 
multi-reduction and generic NP-hardness to strong NP-hardness.
All definitions and proofs are in \Cref{sec:strong-generic-hardness}.

\section{Conclusion}

Our work presented several results related to  fairly allocating objects among agents, with a mixture of divisible and indivisible objects. In our work all objects are divisible, however, divisibility is highly discouraged. We covered many fairness and efficiency concepts, with a bound on the number of sharings or shared objects. 
We tackle the difficulty of finding an allocation for agents with arbitrary valuations by restricting the agent valuations to some well-studied domains, especially the \emph{binary}, \emph{generalized-binary}, and the \emph{non-degenerate} valuations. 
In addition, our work shows that sometimes a polynomial algorithm can be designed if we allow a bounded number of objects to be shared among agents. Such a behavior is already analyzed in \citet{DBLP:journals/corr/abs-2204-11753}, and our work can be seen as a confirmation --- when searching for a fair allocation is too hard, considering the objects as divisible, but still constraining the number of shared objects, might be reasonable for the agents, and may significantly decrease the runtime of the search. Thus, we highlighted a new way to relax the difficult fair-division problem, as it is common in the literature using for example Envy-freeness up to one good (EF1).

Our main result considers 3 agents under equal-sum generalized binary utilities. 
Despite the fact that we consider only 3 agents, the algorithm is complicated, and involves an exhaustive case analysis. That is why, as a next challenge, one can try to generalize our algorithm and see if it is possible to extend it for any fixed number $n$ of agents. Also, it is interesting to see if the same result holds for agents under general additive utilities.

We dedicated one section to non-degenerate valuations. Despite the surprising polynomial time algorithm designed in \citet{sandomirskiy2022efficient}, several results in our paper point to the hardness of this problem, even when allowing shared objects.

Along the paper, we left many open problems. This paper is an invitation for researchers interested in the fair-division field, 
to find new  results, 
with the goal to find theoretic and practical solutions, to numerous fair-division problems. 
\newpage

\begin{APPENDICES}
\SingleSpacedXII 

\section{Truthful Fair Division}

\label{sec:truthful}
A division algorithm is \emph{truthful} if for every agent $i$, the utility of $i$ is maximized when $i$ reports the true valuations  $(v_{i,o})_{o\in[m]}$.
In general, truthfulness, fairness, and Pareto-optimality are incompatible, see 
\citet{zhou1990conjecture}. However, truthfulness can be achieved by introducing some inefficiencies.
This section surveys several such truthful mechanisms and checks whether any of them can be adapted to minimize the sharing.

\subsection{Consensus allocation mechanism}
\citet{Mossel2010Truthful} presents a truthful randomized mechanism that uses a consensus allocation.
Given a consensus allocation $\mathbf{z}$,
a permutation $\pi$ over $[n]$
is selected uniformly at random,
and the bundle $\mathbf{z}_{\pi(i)}$ is allocated to agent $i$. Thus, the expected utility of any agent $i$, whether truthful or not, is $V_i/n$, so the agent cannot gain by false reporting. 
Moreover, a truthful agent gets a utility of exactly $V_i/n$ with certainty, while a non-truthful agent might get more or less than $V_i/n$. So for a risk-averse agent, truthfulness is a strictly dominant strategy.

Combining the algorithm of \citet{Mossel2010Truthful} with Theorem \ref{thm:upperbound-cons} gives: 
\begin{corollary}
There exists a randomized truthful algorithm that (with certainty) returns an envy-free allocation with  at most $n(n-1)$ sharings.
\end{corollary}
The $n(n-1)$ is tight for a consensus allocation, but not necessarily tight for truthfulness. Below we survey some other approaches for truthful fair allocation.

\subsection{Partial-allocation and strong-demand mechanisms}
\citet{cole2013mechanism} suggest a different approach to truthful fair division, called \emph{Partial Allocation  Mechanism (PAM)}.
Their benchmark for fairness is the \emph{max-product allocation} --- the allocation maximizing the product of utilities (also known as the \emph{Nash-optimal} or the \emph{Proportionally-fair} allocation).
Informally, PAM works as follows.
\begin{enumerate}
\item Find a max-product allocation $\mathbf{z^*}$.
\item For each agent $i$, compute the ratio $f_i$ between the product of the other agents' utilities when $i$ is present, and the maximum product of their utilities when $i$ is not present. 
\item Give to each agent $i$ the bundle 
$f_i\cdot \mathbf{z^*_i}$ (that is,  a fraction $f_i$ of each good $i$ receives in the original max-product allocation).
\end{enumerate}
They prove that, under reasonable assumptions, $f_i \geq 1/e \approx 0.368$, so each agent is guaranteed at least $36.8\%$ of his/her max-product utility.  Moreover, they prove that, with additive linear utilities, the allocation is envy-free.
It remains to analyze how many sharings this allocation requires.

It is well-known that, when all valuations are linear and additive as in our case, the max-product allocation is equivalent to the Competitive Equilibrium from Equal Incomes (CEEI) allocation which is a well-known rule for fair allocation of resources among agents with different preferences.
As explained in Section \ref{sec:bounds}, there always exists a CEEI allocation with at most $n-1$ sharings. However, 
the process of giving each agent $i$ only a fraction $f_i$ of each object is, arguably, equivalent to forcing the agent to share each and every object (e.g. with the mechanism designer, the government or the public).

A possible way to overcome this issue is to interpret $f_i$ not as a fraction but as a probability, that is: with probability $f_i$, agent $i$ receives his/her share in the max-product allocation, otherwise agent $i$ receives nothing. Since the utilities of all agents (assuming they are risk-neutral) are the same in both cases, the randomized mechanism is still truthful and has at most $n-1$ sharings. However, in contrast to the randomized mechanism based on consensus allocation, it does not guarantee proportionality ex-post, nor even ex-ante (since $1/e$ of the max-product allocation may be less than $1/n$ of the total value). \citet{cole2013mechanism} presents a different mechanism called Strong Demand Mechanism (SDM), which is particularly efficient when there are many agents and few goods (as in the case of coupon-based privatization). SDM gives to each agent a fraction of a single good, so all objects are shared. Hence, SDM is apparently incompatible with sharing minimization.

\begin{remark}
For the case of two agents, \citet{cole2013positive} present a mechanism called MAX, which runs both the consensus-allocation mechanism and the partial-allocation mechanism and returns the outcome with the highest social welfare. They prove that the resulting mechanism is truthful, envy-free, proportional,  and guarantees at least $2/3$ of the maximum social welfare. 
Moreover, as explained above, for two agents the consensus-allocation mechanism can be implemented with at most 2 sharings and the partial-allocation mechanism can be implemented with at most 1 sharing.
\end{remark}

\subsection{Mechanism for binary valuations}
In the problem of fair cake-cutting, most positive results for truthful allocation are designed for the special case in which the agents have \emph{piecewise-uniform valuations}.
the cake is assumed to be an interval;
each agent $i$ desires a finite set of sub-intervals of the cake, and does not care about the rest of the cake.
For this setting, \citet{Chen2013Truth} presents a truthful mechanism, called the CLPP mechanism, that guarantees envy-freeness and Pareto-optimality. Informally, the CLPP mechanism proceeds as follows.
\begin{enumerate}
\item For each subset $S$ of agents, calculate their \emph{average size} --- the total size of intervals desired by at least one member of $S$.
\item Pick an $S$ with a smallest average size (breaking ties arbitrarily); 
\item \label{step:allocate} Allocate to the members of $S$, all their desired sub-intervals, such that each member of $S$ gets a value of exactly the average size (this means that all the desired sub-intervals are allocated to agents who desire them).
\item Divide the remaining cake recursively among the remaining agents.
\item Discard all parts of the cake that are undesired by any agent.
\end{enumerate}
In our setting, piecewise-uniform valuations means that each agent assigns to each object a value of either $1$ or $0$. 
Under this assumption, an agent cares only about the total amount he/she gets from desired objects. Therefore, in step \ref{step:allocate}, the allocation can be implemented using at most $|S|-1$ sharings. When $S$ is the set of all agents, the allocation requires $n-1$ sharings, which is the worst-case upper bound on any fair division.

\begin{corollary}
For binary valuations, 
there exists a deterministic truthful algorithm that (with certainty) returns an envy-free and Pareto-optimal allocation with  at most $n-1$ sharings.
\end{corollary}

\begin{open}
(a)
In some situations, it may be possible to attain the same allocation as the CLPP mechanism with less than $n-1$ sharings. Is it possible to minimize the amount of sharing, while keeping the mechanism truthful, envy-free and Pareto-optimal?

(b) What is the smallest number of sharings required for a truthful mechanism for finding proportional / envy-free / equitable allocation?
\end{open}


\section{Generalized Binary Utilities -- Cases Details}\label{app:g-binary}

In this section, we handle the 3 cases to prove \Cref{thm:g-binary}.
We consider three cases:
\begin{itemize}
    \item $\sumgb{1} \leq \sumgb{4}$ --- there is always a PROP allocation (\ref{sec:1-smaller-4});
    \item $\sumgb{4} < \sumgb{1} \leq \frac{2}{3} \cdot V$ --- there is always a PROP allocation  (\ref{sec:4-smaller-1});
    \item $\sumgb{1} > \frac{2}{3} \cdot V$ --- there is \emph{not} always a PROP allocation (\ref{sec:1-big}).
\end{itemize}

\end{APPENDICES}

\begin{alphasection}

\subsubsection{$\sumgb{1} \leq \sumgb{4}$, there is always a PROP allocation} \label{sec:1-smaller-4}
~
\begin{table}[H]
    \centering
    \begin{tabular}{c || c || c || c }
       ${\cal X}_j$ & $\sum_{x_i \in {\cal X}_j} u_1(x_i)$ & $\sum_{x_i \in {\cal X}_j} u_2(x_i)$ & $\sum_{x_i \in {\cal X}_j} u_3(x_i)$ \\
       \hline
       \hline
       ${\cal X}_1$  & $\sumgb{1}$ & $\sumgb{1}$ & $\sumgb{1}*$\\
       \hline
       ${\cal X}_2$  & $\sumgb{2}$ & $\sumgb{2}*$ & 0 \\
       \hline
       ${\cal X}_3$  & $\sumgb{3}*$ & 0 & $\sumgb{3}$\\
       \hline
       ${\cal X}_4$  & 0 & $\sumgb{4} \longrightarrow \frac{\sumgb{1} + \sumgb{4}}{3}*$ & $\sumgb{4} \longrightarrow \frac{2 \cdot \sumgb{4} - \sumgb{1}}{3}*$\\
       \hline
       ${\cal X}_5$  & $\sumgb{5}*$ & 0 & 0\\
       \hline
       ${\cal X}_6$  & 0 & $\sumgb{6}*$ & 0\\
       \hline
       ${\cal X}_7$  & 0 & 0 & $\sumgb{7}*$\\
       \hline
       bundle sum  & $\sumgb{3}+\sumgb{5}$ & $\sumgb{2}+\frac{\sumgb{1} + \sumgb{4}}{3}+\sumgb{6}$ & $\sumgb{1}+ \frac{2 \cdot \sumgb{4} - \sumgb{1}}{3}+\sumgb{7}$\\
    \end{tabular}
    \label{tab:1}
\end{table}

We assign the objects as shown in the table, such that there is at most one shared object in ${\cal X}_4$.
The intuitive meaning of this case is that a set of objects is worth a lot to agents 2 and 3, but worth little to agent 1. We can allocate these objects among agents 2 and 3 in any way we like with a single shared object, and use the other objects (which are less valuable) to compensate agent 1 without additional splits. Recall that by assumption $\sumgb{2} \leq \sumgb{3} \leq \sumgb{4}$ and $\sumgb{1} \leq \sumgb{4}$. Formally, the allocation is PROP since,

\begin{itemize}
    \item $\sumgb{3}+\sumgb{5} = \sumgb{4} + \sumgb{6} \geq \frac{\sumgb{4} + \sumgb{4} + \sumgb{4} + \sumgb{6}}{3} \geq \frac{\sumgb{1} + \sumgb{2} + \sumgb{4} + \sumgb{6}}{3} = \frac{V}{3}$;
    \item $\sumgb{2}+\frac{\sumgb{1} + \sumgb{4}}{3}+\sumgb{6}  \geq  \frac{\sumgb{1}+\sumgb{2}+\sumgb{4}+\sumgb{6}}{3} = \frac{V}{3}$;
    \item $\sumgb{1} + \frac{2 \cdot \sumgb{4} - \sumgb{1}}{3}+\sumgb{7} \geq \frac{\sumgb{1}+\sumgb{4}+\sumgb{4}+\sumgb{7}}{3} \geq \frac{\sumgb{1}+\sumgb{3}+\sumgb{4}+\sumgb{7}}{3} = \frac{V}{3}$.
\end{itemize}

\subsubsection{$\sumgb{4} < \sumgb{1} \leq \frac{2}{3} \cdot V$, there is always a PROP allocation} \label{sec:4-smaller-1}
~
\begin{table}[H]
    \centering
    \begin{tabular}{c || c || c || c }
       ${\cal X}_j$ & $\sum_{x_i \in {\cal X}_j} u_1(x_i)$ & $\sum_{x_i \in {\cal X}_j} u_2(x_i)$ & $\sum_{x_i \in {\cal X}_j} u_3(x_i)$ \\
       \hline
       \hline
       ${\cal X}_1$  & $\sumgb{1}$ & $\sumgb{1} \longrightarrow \sumgb{1}'*$ & $\sumgb{1} \longrightarrow \sumgb{1}''*$\\
       \hline
       ${\cal X}_2$  & $\sumgb{2}*$ & $\sumgb{2}$ & 0 \\
       \hline
       ${\cal X}_3$  & $\sumgb{3}*$ & 0 & $\sumgb{3}$\\
       \hline
       ${\cal X}_4$  & 0 & $\sumgb{4} \longrightarrow \sumgb{4}'*$ & $\sumgb{4} \longrightarrow  \sumgb{4}''*$\\
       \hline
       ${\cal X}_5$  & $\sumgb{5}*$ & 0 & 0\\
       \hline
       ${\cal X}_6$  & 0 & $\sumgb{6}*$ & 0\\
       \hline
       ${\cal X}_7$  & 0 & 0 & $\sumgb{7}*$\\
       \hline
       bundle sum  & $\sumgb{2}+\sumgb{3}+\sumgb{5}$ & $\sumgb{1}'+\sumgb{4}'+\sumgb{6}$ & $\sumgb{1}''+\sumgb{4}''+\sumgb{7}$\\
    \end{tabular}
    \label{tab:2}
\end{table}

We assign the objects as shown in the table, such that there is at most one shared object in ${\cal X}_1 \cup {\cal X}_4$. We combine ${\cal X}_1$ and ${\cal X}_4$, and cut at half such that $\sumgb{1}'+\sumgb{4}'=\sumgb{1}''+\sumgb{4}''=\frac{\sumgb{1}+\sumgb{4}}{2}$.
The intuitive meaning of this case is that there is a set of objects that are worth a lot to all three agents, but the sum of these objects is at most $2/3$ of the total value $V$. Therefore, it is enough to allocate this set among two of the three agents (giving each of them at least $V/3$), and use the other objects to compensate the third agent.
Recall that by assumption $\sumgb{2} \leq \sumgb{3} \leq \sumgb{4} < \sumgb{1}$, and, $\sumgb{1} \leq \frac{2}{3} \cdot V$. Formally, the allocation is PROP since,
\begin{itemize}
    \item $\sumgb{2}+\sumgb{3}+\sumgb{5} = \sumgb{1} + \sumgb{2} + \sumgb{3} + \sumgb{5} - \sumgb{1}= V - \sumgb{1} \geq V - \frac{2}{3} \cdot V = \frac{V}{3}$;
    \item $\sumgb{1}'+\sumgb{4}'+\sumgb{6} = \frac{\sumgb{1}+\sumgb{4}}{2} + \sumgb{6} = \frac{\sumgb{1}+\sumgb{4}}{3} +  \frac{\sumgb{1}+\sumgb{4}}{6} + \sumgb{6}  \geq \frac{\sumgb{1}+\sumgb{4} + \sumgb{6}}{3} +  \frac{\sumgb{2}+\sumgb{2}}{6}  \geq  \frac{\sumgb{1}+\sumgb{2}+\sumgb{4}+\sumgb{6}}{3} = \frac{V}{3}$;
    \item $\sumgb{1}''+\sumgb{4}''+\sumgb{7} = \frac{\sumgb{1}+\sumgb{4}}{2} + \sumgb{7} = \frac{\sumgb{1}+\sumgb{4}}{3} +  \frac{\sumgb{1}+\sumgb{4}}{6} + \sumgb{7}  \geq \frac{\sumgb{1}+\sumgb{4} + \sumgb{7}}{3} +  \frac{\sumgb{3}+\sumgb{3}}{6}  \geq  \frac{\sumgb{1}+\sumgb{3}+\sumgb{4}+\sumgb{7}}{3} = \frac{V}{3}$.
\end{itemize}

\subsubsection{$\sumgb{1} > \frac{2}{3} \cdot V$} \label{sec:1-big}

The intuitive meaning of this case is that there is a set of objects that are worth a lot to all three agents, and their sum is larger than $2/3$ the total value. Therefore, we must partition these objects (objects in ${\cal X}_1$) among all three agents, using \maxminnRwayidenticalsplitpar{3,1}.
We run \maxminnRwayidenticalsplitpar{3,1} on objects in ${\cal X}_1$, and we denote the resulting bundles sum by $b_1, b_2, b_3$, where without loss of generality, $b_3$ is the largest bundle sum.
By \Cref{lem:maxminstructure}, $b_1=b_2\leq b_3$.

We consider three cases,
\begin{itemize}    
    \item $b_3 \leq \frac{\sumgb{1}+\sumgb{3}+\sumgb{4}-2 \cdot \sumgb{7}}{3}$;
    \item $\frac{\sumgb{1}+\sumgb{3}+\sumgb{4}-2 \cdot \sumgb{7}}{3} < b_3 \leq \frac{\sumgb{1}+\sumgb{2}+2 \cdot (\sumgb{3}+\sumgb{4}+\sumgb{5}+\sumgb{6})}{3}$;
    \item $b_3 > \frac{\sumgb{1}+\sumgb{2}+2 \cdot (\sumgb{3}+\sumgb{4}+\sumgb{5}+\sumgb{6})}{3}$.
\end{itemize}

\paragraph{$b_3 \leq \frac{\sumgb{1}+\sumgb{3}+\sumgb{4}-2 \cdot \sumgb{7}}{3}$, there is always a PROP allocation}
\begin{table}[H]
    \centering
    \begin{tabular}{c || c || c || c }
       ${\cal X}_j$ & $\sum_{x_i \in {\cal X}_j} u_1(x_i)$ & $\sum_{x_i \in {\cal X}_j} u_2(x_i)$ & $\sum_{x_i \in {\cal X}_j} u_3(x_i)$ \\
       \hline
       \hline
       ${\cal X}_1$  & $\sumgb{1} \longrightarrow \sumgb{1}'*$ & $\sumgb{1} \longrightarrow \sumgb{1}''*$ & $\sumgb{1} \longrightarrow \sumgb{1}'''*$\\
       \hline
       ${\cal X}_2$  & $\sumgb{2}$ & $\sumgb{2}*$ & 0 \\
       \hline
       ${\cal X}_3$  & $\sumgb{3}*$ & 0 & $\sumgb{3}$\\
       \hline
       ${\cal X}_4$  & 0 & $\sumgb{4}$ & $\sumgb{4}*$\\
       \hline
       ${\cal X}_5$  & $\sumgb{5}*$ & 0 & 0\\
       \hline
       ${\cal X}_6$  & 0 & $\sumgb{6}*$ & 0\\
       \hline
       ${\cal X}_7$  & 0 & 0 & $\sumgb{7}*$\\
       \hline
       bundle sum  & $\sumgb{1}' + \sumgb{3}+\sumgb{5}$ & $\sumgb{1}''+\sumgb{2}+\sumgb{6}$ & $\sumgb{1}'''+\sumgb{4}+\sumgb{7}$\\
    \end{tabular}
    \label{tab:3}
    \end{table}
We assign the objects as shown in the table. The intuitive meaning of this is case is that the partition of $\sumgb{1}$ is quite balanced, so we can make it PROP by allocating the other objects.

Denote the value of the most valuable object in $\sumgb{1}$ by $3M$.
\subparagraph{$b_1 = b_2 = b_3$ and $\sumgb{4} < 3M$}

By \Cref{lem:largestItems}, the shared object is the most valuable one (with value $3M$), by the pigeonhole principle, there is at least one bundle in which the part of the shared object is at least $M$. Give this bundle to agent 3.

We have $\frac{\sumgb{4}}{3} < M$, so, starting from the perfect allocation, we can move $\frac{\sumgb{4}}{3}$ from the shared object from agent 3 to agent 2. Therefore, $\sumgb{1}'=\frac{\sumgb{1}}{3}$, $\sumgb{1}''=\frac{\sumgb{1}+\sumgb{4}}{3}$ and $\sumgb{1}'''=\frac{\sumgb{1}-\sumgb{4}}{3}$.

The allocation is PROP since,
\begin{itemize}
    \item $\sumgb{1}' + \sumgb{3}+\sumgb{5} = \frac{\sumgb{1}}{3} + \sumgb{3}+\sumgb{5} \geq \frac{\sumgb{1} + \sumgb{2} + \sumgb{3} + \sumgb{5}}{3} = \frac{V}{3}$;
    \item $\sumgb{1}''+\sumgb{2}+\sumgb{6} = \frac{\sumgb{1}+\sumgb{4}}{3}+\sumgb{2}+\sumgb{6} \geq \frac{\sumgb{1} + \sumgb{2} + \sumgb{4} + \sumgb{6}}{3} = \frac{V}{3}$;
    \item $\sumgb{1}'''+\sumgb{4}+\sumgb{7} = \frac{\sumgb{1}-\sumgb{4}}{3}+\sumgb{4}+\sumgb{7} \geq \frac{\sumgb{1}+\sumgb{3}+\sumgb{4}+\sumgb{7}}{3} = \frac{V}{3}$.
\end{itemize}

\subparagraph{$b_1 = b_2 = b_3$ and $\sumgb{4} \geq 3M$}\label{sec:big_V4}

Since the largest object in $\sumgb{1}$ is equal to $3M$, we define $\sumgb{1}'''$ as follows. Cut from ${\cal X}_1$ a subset with value $\frac{\sumgb{1}+\sumgb{3}- 2 \cdot (\sumgb{4}+\sumgb{7})}{3}$. The cut is feasible since:
\begin{align}\label{eq:V_1}
    \notag
    \sumgb{1} > \frac{2}{3} \cdot V & \iff 3 \cdot \sumgb{1} > 2 \cdot V \\
    \notag
     & \iff 3 \cdot \sumgb{1} > 2 \cdot (\sumgb{1}+\sumgb{3}+\sumgb{4}+\sumgb{7}) \\
    & \iff \sumgb{1} > 2 \cdot (\sumgb{3}+\sumgb{4}+\sumgb{7});
\end{align}
so $0 < \frac{\sumgb{1}+\sumgb{3}- 2 \cdot (\sumgb{4}+\sumgb{7})}{3} \leq \sumgb{1}$. We add the shared object (if any) to this subset. Now there is no shared object in $\sumgb{1}'''$, and 
\begin{align*}
     & \frac{\sumgb{1}+\sumgb{3}- 2 \cdot   (\sumgb{4}+\sumgb{7})}{3} 
    \leq \sumgb{1}''' 
    \\ & \hspace{3cm} \leq \frac{\sumgb{1}+\sumgb{3}- 2 \cdot (\sumgb{4}+\sumgb{7})}{3} + 3M 
    \\ & \hspace{3cm} \leq \frac{\sumgb{1}+\sumgb{3}- 2 \cdot (\sumgb{4}+\sumgb{7})}{3} + \sumgb{4} 
    \\ & \hspace{3cm} =\frac{\sumgb{1}+\sumgb{3}+\sumgb{4}- 2 \cdot \sumgb{7}}{3}.
\end{align*}
The sum of the remaining objects is at least:
\begin{align*}
     \sumgb{1} - \sumgb{1}''' & \geq \sumgb{1} - \frac{\sumgb{1}+\sumgb{3}+\sumgb{4}- 2 \cdot \sumgb{7}}{3} 
     \\ & = \frac{2 \cdot (\sumgb{1} + \sumgb{7}) - \sumgb{3}-\sumgb{4}}{3}.
\end{align*}
Since there is no shared object in $\sumgb{1}'''$ we can cut the remaining objects and set $\sumgb{1}'' = \frac{\sumgb{1} + \sumgb{4} - 2 \cdot (\sumgb{3}+\sumgb{7})}{3}$.
The cut is feasible since:
\begin{align*}
    \sumgb{1} - \sumgb{1}''' 
    & \geq \frac{2 \cdot (\sumgb{1} + \sumgb{7}) - \sumgb{3}-\sumgb{4}}{3} 
    \\ & \geq \frac{\sumgb{1} + \sumgb{1} - \sumgb{3}-\sumgb{4}}{3} 
    \\ &
    \stackrel{\eqref{eq:V_1}}{>} 
    \frac{\sumgb{1} 
    + 2 \cdot (\sumgb{3}+\sumgb{4}+\sumgb{7}) - 
    \sumgb{3}-\sumgb{4}}{3} 
    \\ & \geq \frac{\sumgb{1} + \sumgb{4} - 2 \cdot (\sumgb{3}+\sumgb{7})}{3};
\end{align*}
so $0 < \frac{\sumgb{1} + \sumgb{4} - 2 \cdot (\sumgb{3}+\sumgb{7})}{3} \leq \sumgb{1} - \sumgb{1}'''$.
Also, $\sumgb{1}' \geq \sumgb{1} - \sumgb{1}''' - \sumgb{1}''$, which is at least:
\begin{align*}
    \sumgb{1}' & \geq \frac{2 \cdot (\sumgb{1} + \sumgb{7}) - \sumgb{3}-\sumgb{4}}{3} 
    \\ & \hspace{3 cm} 
    - \frac{\sumgb{1} + \sumgb{4} - 2 \cdot (\sumgb{3}+\sumgb{7})}{3} 
    \\ & = \frac{\sumgb{1}+\sumgb{3}-2\cdot \sumgb{4} + 4 \cdot \sumgb{7}}{3}.
\end{align*}
The allocation is PROP since,
\begin{itemize}
    \item $\sumgb{1}' + \sumgb{3}+\sumgb{5} \geq \frac{\sumgb{1}+\sumgb{3}-2\cdot \sumgb{4} + 4 \cdot \sumgb{7}}{3} + \sumgb{3}+\sumgb{5} 
    \stackrel{(\sumgb{3} + \sumgb{5} = \sumgb{4} + \sumgb{6})}{\geq} 
    \frac{\sumgb{1}+\sumgb{3}-2\cdot \sumgb{4} + \sumgb{7}}{3} + \sumgb{4} \geq \frac{\sumgb{1} + \sumgb{3} + \sumgb{4} + \sumgb{7}}{3} = \frac{V}{3}$;
    \item $\sumgb{1}''+\sumgb{2}+\sumgb{6} \geq \frac{\sumgb{1} + \sumgb{4} - 2 \cdot (\sumgb{2}+\sumgb{6})}{3}+\sumgb{2}+\sumgb{6} \geq \frac{\sumgb{1} + \sumgb{2} + \sumgb{4} + \sumgb{6}}{3} = \frac{V}{3}$;
    \item $\sumgb{1}'''+\sumgb{4}+\sumgb{7} \geq \frac{\sumgb{1}+\sumgb{3}- 2 \cdot (\sumgb{4}+\sumgb{7})}{3} + \sumgb{4} + \sumgb{7} \geq \frac{\sumgb{1}+\sumgb{3}+\sumgb{4}+\sumgb{7}}{3} = \frac{V}{3}$.
\end{itemize}

\subparagraph{$b_1 = b_2 < b_3$}

Set $\sumgb{1}'''=b_3$.
By assumption, $\frac{\sumgb{1}}{3} < b_3 = \sumgb{1}''' \leq \frac{\sumgb{1}+\sumgb{3}+\sumgb{4}-2 \cdot \sumgb{7}}{3}$, and by \Cref{lem:maxminstructure}, there is no shared object in $\sumgb{1}'''$. So, we reach the exact same structure from \ref{sec:big_V4}. Using the same arguments the allocation is PROP.

\paragraph{$\frac{\sumgb{1}+\sumgb{3}+\sumgb{4}-2 \cdot \sumgb{7}}{3} \leq b_3 \leq \frac{\sumgb{1}+\sumgb{2}+2 \cdot (\sumgb{3}+\sumgb{4}+\sumgb{5}+\sumgb{6})}{3}$, there is always a PROP allocation}
\begin{table}[H]
    \centering
    \begin{tabular}{c || c || c || c }
       ${\cal X}_j$ & $\sum_{x_i \in {\cal X}_j} u_1(x_i)$ & $\sum_{x_i \in {\cal X}_j} u_2(x_i)$ & $\sum_{x_i \in {\cal X}_j} u_3(x_i)$ \\
       \hline
       \hline
       ${\cal X}_1$  & $\sumgb{1} \longrightarrow \sumgb{1}'*$ & $\sumgb{1} \longrightarrow  \sumgb{1}''*$ & $\sumgb{1} \longrightarrow  \sumgb{1}'''*$\\
       \hline
       ${\cal X}_2$  & $\sumgb{2} \longrightarrow \sumgb{2}'*$ & $\sumgb{2} \longrightarrow \sumgb{2}''*$ & 0 \\
       \hline
       ${\cal X}_3$  & $\sumgb{3}*$ & 0 & $\sumgb{3}$\\
       \hline
       ${\cal X}_4$  & 0 & $\sumgb{4}*$ & $\sumgb{4}$\\
       \hline
       ${\cal X}_5$  & $\sumgb{5}*$ & 0 & 0\\
       \hline
       ${\cal X}_6$  & 0 & $\sumgb{6}*$ & 0\\
       \hline
       ${\cal X}_7$  & 0 & 0 & $\sumgb{7}*$\\
       \hline
       bundle sum  & $  \sumgb{1}' + \sumgb{2}' + \sumgb{3}+\sumgb{5}$ & $ \sumgb{1}''+\sumgb{2}'' + \sumgb{4}+\sumgb{6}$ & $ \sumgb{1}'''+\sumgb{7}$ \\
    \end{tabular}
    \label{tab:4}
\end{table}
We assign the objects as shown in the table. This case is an intermediate case, where the partition of $\sumgb{1}$ is not balanced, but the maximum is still not very high. Then, we can satisfy one of the agents by giving him the largest part of $\sumgb{1}$, and compensate the other agents using the remaining objects.

Set $\sumgb{1}'''=b_3$.
By assumption, $b_3=\sumgb{1}''' \leq \frac{\sumgb{1}+\sumgb{2}+2 \cdot (\sumgb{3}+\sumgb{4}+\sumgb{5}+\sumgb{6})}{3}$, so the sum of the remaining objects is at least:
\begin{align*}
    \sumgb{1}'&+\sumgb{1}''+\sumgb{2}'+\sumgb{2}'' = \sumgb{1}-\sumgb{1}'''+\sumgb{2} 
    \\ & \geq  \sumgb{1}-\frac{\sumgb{1}+\sumgb{2}+2 \cdot (\sumgb{3}+\sumgb{4}+\sumgb{5}+\sumgb{6})}{3} +\sumgb{2} 
    \\ & = \frac{2\cdot(\sumgb{1}+\sumgb{2})-2 \cdot (\sumgb{3}+\sumgb{4}+\sumgb{5}+\sumgb{6})}{3}.
\end{align*}
By \Cref{lem:maxminstructure}, there is no shared object in $\sumgb{1}'''$, so we can cut the remaining objects and set 
$\sumgb{1}'+\sumgb{2}'=\frac{\sumgb{1}+\sumgb{2}-2 \cdot (\sumgb{3}+\sumgb{5})}{3}$, and the remainder $\sumgb{1}''+\sumgb{2}'' \geq \frac{\sumgb{1}+\sumgb{2}-2 \cdot (\sumgb{4}+\sumgb{6})}{3}$.
The allocation is PROP since,
\begin{itemize}
    \item $\sumgb{1}' + \sumgb{2}' + \sumgb{3}+\sumgb{5} \geq \frac{\sumgb{1}+\sumgb{2}-2 \cdot (\sumgb{3}+\sumgb{5})}{3} + \sumgb{3}+\sumgb{5} = \frac{\sumgb{1} + \sumgb{2} + \sumgb{3} + \sumgb{5}}{3} = \frac{V}{3}$;
    \item $\sumgb{1}'' + \sumgb{2}'' + \sumgb{4}+\sumgb{6} \geq \frac{\sumgb{1}+\sumgb{2}-2 \cdot (\sumgb{4}+\sumgb{6})}{3} + \sumgb{4}+\sumgb{6} = \frac{\sumgb{1} + \sumgb{2} + \sumgb{4} + \sumgb{6}}{3} = \frac{V}{3}$;
    \item $\sumgb{1}'''+\sumgb{7} \geq \frac{\sumgb{1}+\sumgb{3}+\sumgb{4}-2 \cdot \sumgb{7}}{3} + \sumgb{7} = \frac{\sumgb{1}+\sumgb{3}+\sumgb{4}+\sumgb{7}}{3} = \frac{V}{3}$.
\end{itemize}

\paragraph{$b_3 > \frac{\sumgb{1}+\sumgb{2}+2 \cdot (\sumgb{3}+\sumgb{4}+\sumgb{5}+\sumgb{6})}{3}$, there is \emph{no} PROP allocation}

In this case, the optimal partition of $\sumgb{1}$ is so unbalanced, that it is not possible to make it PROP using the remaining objects. We prove it using the following Lemma, which claims that if a PROP allocation exists, then, $b_3$ must be at most $\frac{\sumgb{1}+\sumgb{2}+2 \cdot (\sumgb{3}+\sumgb{4}+\sumgb{5}+\sumgb{6})}{3}$, which is infeasible in our current case.

\begin{lemma}
    If there exists a PROP allocation, then the algorithm for \maxminnRwayidenticalsplitpar{3, 1}  must find a partition with   $b_3 \leq \frac{\sumgb{1} + \sumgb{2} + 2 \cdot (\sumgb{3} + \sumgb{4} + \sumgb{5} + \sumgb{6})}{3}$.
\end{lemma}
\proof{Proof}
Suppose there exists a PROP allocation. 
For each agent $i$ and subset ${\cal X}_j$, denote by $F_i(\sumgb{j})$  the total value of objects assigned to $i$ from subset ${\cal X}_j$ in that allocation.
We now show that, for all $i\in\{1,2,3\}$,
$F_i(\sumgb{1})\leq \frac{\sumgb{1} + \sumgb{2} + 2 \cdot (\sumgb{3} + \sumgb{4} + \sumgb{5} + \sumgb{6})}{3}$.

First assume that $F_3(\sumgb{1}) \geq F_2(\sumgb{1}), F_1(\sumgb{1})$.
    By PROP definition, we know that:
    \begin{align*}
        F_1(\sumgb{1}) + F_1(\sumgb{2}) + & F_1(\sumgb{3}) + \sumgb{5} 
        \\ & \geq \frac{\sumgb{1}+\sumgb{2}+\sumgb{3}+\sumgb{5}}{3}; \\ 
        F_2(\sumgb{1}) + F_2(\sumgb{2}) + & F_2(\sumgb{4}) + \sumgb{6} 
        \\ & \geq \frac{\sumgb{1}+\sumgb{2}+\sumgb{4}+\sumgb{6}}{3}.
    \end{align*}
    So,
    \begin{align*}
        F_1(\sumgb{1})  \geq & \frac{\sumgb{1}+\sumgb{2}+\sumgb{3} - 2 \cdot \sumgb{5}}{3} 
        \\ & \hspace{2 cm} - F_1(\sumgb{2}) - F_1(\sumgb{3}); \\
        F_2(\sumgb{1})  \geq & \frac{\sumgb{1}+\sumgb{2}+\sumgb{4}- 2 \cdot \sumgb{6}}{3} 
        \\ & \hspace{2 cm} - F_2(\sumgb{2}) - F_2(\sumgb{4}).
    \end{align*}
    Summing up these inequalities gives:
    \begin{align*}
        F_1(\sumgb{1}) + F_2(\sumgb{1}) & \geq \frac{\sumgb{1}+\sumgb{2}+\sumgb{3}- 2 \cdot \sumgb{5}}{3} 
        \\ & \hspace{1cm}  - F_1(\sumgb{2}) - F_1(\sumgb{3}) 
        \\ & \hspace{1cm}  + \frac{\sumgb{1}+\sumgb{2}+\sumgb{4}- 2 \cdot \sumgb{6}}{3} 
        \\ & \hspace{1cm}  - F_2(\sumgb{2}) - F_2(\sumgb{4}) \\
        & \hspace{-2cm} \geq \frac{2 \cdot (\sumgb{1}+\sumgb{2}) + \sumgb{3} + \sumgb{4} - 2 \cdot (\sumgb{5} + \sumgb{6})}{3} 
        \\ &  - \sumgb{2} - \sumgb{3} - \sumgb{4} \\
        & \hspace{-2cm} = \frac{2 \cdot \sumgb{1} - \sumgb{2} - 2 \cdot (\sumgb{3} + \sumgb{4} + \sumgb{5} + \sumgb{6})}{3}.
    \end{align*}
    Since $F_1(\sumgb{1})+F_2(\sumgb{1})+F_3(\sumgb{1})=\sumgb{1}$, we have $F_1(\sumgb{1})+F_2(\sumgb{1})=\sumgb{1}-F_3(\sumgb{1})$, so:
    \begin{align*}
        & \sumgb{1} - F_3(\sumgb{1}) 
        \\ & \hspace{1cm} \geq \frac{2 \cdot \sumgb{1} - \sumgb{2} - 2 \cdot (\sumgb{3} + \sumgb{4} + \sumgb{5} + \sumgb{6})}{3} \\
        & F_3(\sumgb{1}) 
        \\ & \hspace{1cm} \leq \sumgb{1} -  \frac{2 \cdot \sumgb{1} - \sumgb{2} - 2 \cdot (\sumgb{3} + \sumgb{4} + \sumgb{5} + \sumgb{6})}{3} 
        \\ & \hspace{1cm} \leq \frac{\sumgb{1} + \sumgb{2} + 2 \cdot (\sumgb{3} + \sumgb{4} + \sumgb{5} + \sumgb{6})}{3}.
    \end{align*}
    Using a similar logic, if $F_2(\sumgb{1}) \geq F_3(\sumgb{1}), F_1(\sumgb{1})$, we have 
    \begin{align*}
     F_2(\sumgb{1})  & \leq \frac{\sumgb{1} + 2 \cdot \sumgb{2} + \sumgb{3} + 2 \cdot \sumgb{4} + 2 \cdot \sumgb{5} + 2 \cdot \sumgb{7}}{3} \\
     & =  \frac{\sumgb{1} + \sumgb{2} + 2 \cdot (\sumgb{3} + \sumgb{4} + \sumgb{5} + \sumgb{6})}{3} 
     \\ & \hspace{2cm} + \frac{\sumgb{2}+2 \cdot \sumgb{7}-\sumgb{3}-2 \cdot \sumgb{6}}{3}.
    \end{align*}
Since  $\sumgb{2} \leq \sumgb{3}$ and $\sumgb{7} \leq \sumgb{6}$, $F_2(\sumgb{1}) \leq \frac{\sumgb{1} + \sumgb{2} + 2 \cdot (\sumgb{3} + \sumgb{4} + \sumgb{5} + \sumgb{6})}{3}$.

Again, using similar arguments, if $F_1(\sumgb{1}) \geq F_2(\sumgb{1}), F_3(\sumgb{1})$, we have 
    \begin{align*}
     F_1(\sumgb{1})  & \leq \frac{\sumgb{1} + 2 \cdot \sumgb{2} +  2 \cdot \sumgb{3} + \sumgb{4} + 2 \cdot \sumgb{6} + 2 \cdot \sumgb{7}}{3} \\
     & =  \frac{\sumgb{1} + \sumgb{2} + 2 \cdot (\sumgb{3} + \sumgb{4} + \sumgb{5} + \sumgb{6})}{3} 
     \\ & \hspace{2cm} + \frac{\sumgb{2}+ \sumgb{7}-\sumgb{4}-2 \cdot \sumgb{5}}{3}.
    \end{align*}
Since  $\sumgb{2} \leq \sumgb{4}$ and $\sumgb{7} \leq \sumgb{5}$, $F_1(\sumgb{1}) \leq \frac{\sumgb{1} + \sumgb{2} + 2 \cdot (\sumgb{3} + \sumgb{4} + \sumgb{5} + \sumgb{6})}{3}$.

In all three cases, $b_3$ (which we defined as the largest bin sum) satisfies the claimed inequality.
\Halmos \endproof
\end{alphasection}
\begin{APPENDICES}
\setcounter{section}{2}

\section{Strong generic NP-hardness}
\label{sec:strong-generic-hardness}
\SingleSpacedXII 

In many practical problems on integers, it can be assumed that the input integers are not too large; specifically, that the magnitude of each input integer is polynomial (rather than exponential) in the number of inputs. Equivalently, the binary encoding length of each input integer is logarithmic (rather than polynomial) in the number if inputs.

Formally, denote by $\pinputs$ the set of $t$-sized vectors of integers in $[0,2^{ c\cdot \log_2 t}] = [0,t^c]$ for some constant $c$. The following two classic concepts are used to describe run-time complexity of problems restricted to inputs in $\pinputs$:
\begin{itemize}
\item An algorithm is said to run in \emph{pseudo-polynomial-time} if it runs in time poly$(t)$ when restricted to inputs from $\pinputs$;
\item An algorithmic problem on integers is said to be \emph{strongly NP-hard} if it is NP-hard even when restricted to inputs from $\pinputs$.
\end{itemize}
It is well-known that, for any fixed $n$, the problem \textsf{$n$-way Partition} has a pseudo-polynomial-time algorithm (based on dynamic programming).
Similarly, the \fairwithshared{n,s} and \fairwithsharing{n,s} can be solved in pseudo-polynomial time for any fixed $n$ using dynamic programming. 
In contrast, when $n$ is part of the input, these problems are strongly NP-hard, which means that they do not have a pseudo-polynomial-time algorithm unless \textsc{P=NP}.

We would like to extend the notions of generically-easy and generically-hard to problems restricted to $\pinputs$.
The definitions below are almost identical to the analogous definitions from \Cref{sec:non-degenerate}, except that $t$ is replaced with $\log_2 t$.

The following is analogous to \Cref{def:generically-easy}:
\begin{definition}
\label{def:generically-pseudo-easy}
Given an algorithm $A$ for a problem $P$, 
we say that \emph{$A$ runs generically in pseudo-polynomial time}
if there exists a polynomial function $f_p$ such that, for every size $t$ and input $x\in \pinputs$,
there is a subset of ``Good inputs'' $G(x)\subseteq \ball(x,f_p(\log_2 t))$, such that the following holds:

(a) $A$ runs in time poly($t$) on all inputs in $G(x)$;

(b) The fraction of good inputs approaches 1, that is,
\begin{align*}
\lim_{t\to\infty}
\min_{x\in \inputs}\frac{|G(x)|}
{|\ball(x,f_p(\log_2 t))|} = 1    
\end{align*}

(c) Given $x$, it is possible to compute in time poly($t$) a vector  in $G(x)$.
\end{definition}

The following is analogous to \Cref{def:multireduction-full}:
\begin{definition}[strong multi-reduction]
\label{def:multireduction-strong}
Given two decision problems, $P_1$ defined on inputs in $\inputsa$ for some constant $c_1\geq 1$ and $P_2$ defined on inputs in $\inputsb$ for some constant $c_2\geq 1$, 
a \emph{strong multi-reduction} from $P_1$ to $P_2$  consists of a polynomial-time computable function
$t_2: \mathbb{N} \to \mathbb{N}$ 
and 
a family of functions $h_{t_1}: \pinputsa \to \pinputsba$, 
which maps an input for $P_1$ to an input for $P_2$, and satisfies the following:

(a) $h_{t_1}$ runs in time poly$(t_1)$;

(b) There exists an super-polynomial function $f_e$ such that, for all $t_1$ and all $x_1\in \pinputsa$, when $x_2 := h_{t_1}(x_1)$, 
\begin{align*}
P_2(x_2') = P_1(x_1)
&&
\text{ for all }
x_2' \in \ball(x_2, f_e(\log_2 t_2) ).
\end{align*}
\end{definition}
Note the requirement that $h_{t_1}$ map each input to $\pinputsa$ rather than to $\inputsb$; it implies that the reduction must keep the magnitude of the input integers polynomial in the number of inputs.

The following is analogous to \Cref{def:generically-np-hard}:
\begin{definition}
\label{def:generically-strong-np-hard}
A decision problem $P_2$ is called \emph{generically strongly NP hard} if there exists a \emph{strong} multi-reduction from some \emph{strongly} NP-hard problem $P_1$ to $P_2$.
\end{definition}

The following  is analogous to the classic statement that a strongly NP-hard problem does not have a pseudo-polynomial-time algorithm unless \textsc{P=NP}.
Its proof is almost identical to that of  \Cref{thm:generically-easy-is-not-generically-hard}:
\begin{proposition}
If a problem is generically-strongly-NP-hard,
then it does not have a generically-pseudo-polynomial-time algorithm unless \textsc{P=NP}.
\end{proposition}
\proof{Proof}
Let $P_2$ be a decision problem that satisfies the theorem assumptions:

(a) There is a generically-pseudo-polynomial algorithm $A_2$ for $P_2$; denote by $f_p$ the polynomial function in \Cref{def:generically-pseudo-easy}.

(b) There is a multi-reduction $h_{t_1}$ from some NP-hard problem $P_1$ to  $P_2$; denote by $f_e$ the super-polynomial function in \Cref{def:generically-np-hard}.

Since $f_p$ is polynomial and $f_e$ is super-polynomial, there is some $t_0$ such that $f_e(\log_2 t) > f_p(\log_2 t)$ for all $t>t_0$.
We show a polynomial-time algorithm for solving $P_1$ on all inputs of size $t_1$ such that $t_2(t_1)>t_0$.

Given an input $x_1\in \pinputsa$ to problem $P_1$, we
use the assumed multi-reduction $h_{t_1}$ to compute in polynomial time an input $x_2 \in \pinputsb$ to problem $P_2$.

By \Cref{def:generically-pseudo-easy}(a,c), it is possible to compute in time poly$(t_2)$
another input $x_2' \in \ball(x_2, f_p(\log_2 t_2))$, such that $A_2$ runs on $x_2'$ in time poly$(t_2)$.

But by definition of multi-reduction (\Cref{def:multireduction-strong}), 
$P_2(x_2'') = P_1(x_1)$ for all $x_2'' \in \ball(x_2, f_e(\log_2 t_2))$.
In particular, since $f_e(\log_2 t_2) > f_p(\log_2 t_2)$, this holds for $x_2'$ too, so 
$P_2(x_2') = P_1(x_1)$.

Therefore, by running $A_2$ on $x_2'$, we get the correct answer to $P_1(x_1)$.

As $P_1$ is strongly-NP-hard, it is NP-hard even when restricted to inputs in $\pinputsa$.
We have shown that $P_1$ can be solved in polynomial time for all inputs in $\pinputsa$ for sufficiently large $t_1$; hence \textsc{P=NP}.
\Halmos \endproof

We are now ready to extend our hardness results to fair allocation problems in which $n$ is unbounded.

\begin{theorem}
\label{thm:unbounded-n-sharing-nondegenerate}
For every $s \geq 0$
\propwithsharing{s} and \efwithsharing{s} 
are 
generically strongly NP-hard.
\end{theorem}

\proof{Proof}
We show a strong-multi-reduction from $P_1 = $  \textsc{3-Partition}, which is known to be strongly NP-hard.
The input to $P_1$ is a finite multiset $D$ of $t_1 = 3 p$ positive integers $d_1,\ldots,d_{3 p}$,
each of which is in $[0, {t_1}^{6}]$.%
\footnote{
For the strong NP-hardness proof in \citet{garey1979computers}, $c_1=6$ is sufficient.
}
The sum of all integers is $p S$, for some $p\geq 2$ and integer $S$.
The task is to decide if the integers can be partitioned into $p$ triplets such that the sum of each triplet is exactly $S$. 
We denote $L := {t_1}^{6}$.
Note that we can assume w.l.o.g. that each integer in $D$  is in $(S/4,S/2)$, so every set with sum exactly $S$ must be a triplet.

We construct an instance of \fairwithsharing{s} with $n=p+s+1$ agents and $m=3 p+n+1 = 4 p+s+2$ objects. 
Therefore, $t_2 = m n = (p+s+1)(4p + s + 2)$, which is in $O({t_1}^2)$.
The valuations are shown in the table below.

~~~~

\begin{center}
\begin{tabular}{|c||c|c|c|}
\hline
Agent $i$ / Object $o$: 
& 
\shortstack{
$o\in 1,\ldots,3p$ 
\\
(``usual objects'')
}
& 
\shortstack{
$o\in 3p+1,\ldots,3p+n$ 
\\
(``compensation objects'')
}
&
\shortstack{
$o = 3p+n+1$ 
\\
(``special object'')
}
\\
\hline
\hline
\shortstack{$i\in 1,\ldots,p$
\\
(``usual agents'')}
&  
\shortstack{
$L d_o + L S + L/(16 m)$;
}
&  
\shortstack{
$3L/8$ for $o=3p+i$;
\\
$L/(16 m n)$ for $o\neq 3p+i$.
}
&
\shortstack{
$L\cdot (s+1)\cdot 4 S + L/16$;
}
\\
\hline
\shortstack{$i\in p+1,\ldots,n$
\\
(``special agents'')}
 &  
 $L/(16m)$;
 &
  $L/(16m)$;
 &
 \shortstack{
 $L\cdot n \cdot S + L/(16 m)$;
 }
 \\
\hline
\end{tabular}
\end{center}

~~~~

All valuations are in $O(L^2) = O({t_1}^{12})$, so we can let $c_2 := 12$.

For the multi-reduction, we will use the function $f_e(\log t_2) = L/(64 m n) = {t_1}^6 / (64 t_2) \in \Omega({\sqrt{t_2}}^6 / t_2) = \Omega({t_2}^2)$, which is indeed super-polynomial in $\log(t_2)$.

Suppose $D$ has a partition into $p$ triplets, $D_1,\ldots, D_p$, with sum $S$ each.
We give each usual agent $i$ the three usual objects corresponding to the triplet $D_i$, and the compensation object $3p+i$. 
We divide the special object equally among the $s+1$ special agents, using $s$ sharings. 
We show that the resulting allocation is EF (and hence PROP) not only for the constructed valuation matrix, but also for any valuation matrix with infinity-distance of at most $L/(64 m n)$. In all such instances:

- Each usual agent values usual object $o$ at least $L d_o + LS + L/(16m) - L/(64 m n) > L d_o + L S$, and values his compensation object at least $3L/8 - L/(64 m n) > 2L/8$, so he values his own bundle at least $L S + 3 L S + 2L/8 = 4LS+2L/8$.
On the other hand, each usual agent values usual object $o$ at most $L d_o + LS + L/(16m) + L/(64 m n) < L d_o + L S + L/(8m)$, every other compensation object at most $L/(16mn)+L/(64mn) < L/(8mn)$, and the special object at most $L(s+1)\cdot 4S + L/8$, 
so he values the bundle of every other agent (usual or special) at most $4 L S + L/8$. Therefore the usual agents do not envy.

- Each special agent values his own bundle, as well as all special agents' bundles, at exactly the same value, which is at least $(n L S)/(s+1) > L S$.
On the other hand, he values any non-special object at most $L/(8m)$, so values the bundles of usual agents at most $L/8$, so the special agents too do not envy.

Conversely, suppose there is a PROP allocation with at most $s$ sharings. 
Then:

- Each special agent must receive a value larger than $(n L S)/ n = L S$. Since their total value for all objects except the special object is at most $L/8$, each special agent must receive a part of the special object. This means that all $s$ sharings are of the special object, so the usual and compensation objects must be allocated without sharing.

- Each usual agent values all objects together more than 
$4 p L S + 4(s+1)L S = 4 n L S$, so must receive a value larger than $4 L S$. 
The value of each potential bundle without special objects comes from integer multiples of $L$ (--- the $L\cdot d_o + LS$ value of usual objects), and fractions of $L$ (--- the $L/(16 m)$ value of usual objects, and the values of compensation objects). The sum of all these fractions is strictly less than $L$. Since no more sharings are allowed in the special, each usual agent must receive usual objects with no sharings, with total value at least $4 L S$. 

This induces a partition of $D$ into $p$ subsets with sum $\sum_o d_o \geq S$. As the sum of all $D$ is exactly $p S$, the sum of each part in the latter partition must be exactly $S$. As each integer in $D$ is in $(S/4,S/2)$, all parts must be triplets.
\Halmos \endproof

\begin{theorem}\label{thm:shared-strong-np-degenerate}
For any fixed $s\geq 0$, \propwithshared{s} and  \efwithshared{s} are strongly NP-hard.
\end{theorem}
\proof{Proof}
The proof is very similar to that of  \Cref{thm:hardness-fair-with-shared}, but adapted to a strong-multi-reduction from a strongly-NP-hard problem.

We show a multi-reduction from
the case of identical valuations, proved to be strongly-NP-hard in \cite{DBLP:journals/corr/abs-2204-11753}.

An input of size $t_1$ consists of a list $D$ of $t_1-1$ items,
whose values are integers $[d_1,\ldots,d_{t_1-1}]$
in $[0, {t_1}^{c_1}]$, where $c_1$ is some constant, $c_1\geq 6$.
The $t_1$-th input is $n$, which should be at most $t_1$.
The sum of all values is $\sum_{o=1}^{t_1-1} d_o = n S$.
We have to decide if they can be partitioned into $n$ bins with sum $S$ each, with at most $s$ items split between two or more bins.
Let $L := {t_1}^{c_1}$.

We construct an instance of \fairwithshared{s} with $n$ agents and $m = t_1-1+n$ objects. Hence, $t_2 = m n + 1 = n t_1 + n^2 + 1$, which is in $O({t_1}^2)$.
All agents' valuations will be integers in $[0, {t_2}^{c_2}]$, where $c_2 := 2 c_1$.
Recall that ``$/$'' denotes integer division.
\begin{itemize}
\item The objects $o \in \{1,\ldots, t_1-1\}$ are \emph{usual objects}. 
Each agent $i\in[n]$ values each usual object $o$ at $L \cdot d_o + L / (16 m n)$.
\item The objects $o \in \{t_1,\ldots, t_1-1+n\}$ are \emph{compensation objects}. 
Each agent $i\in[n]$ values the compensation object $t_1-1+i$ at 
$ 3 L / (8n)$
and every other compensation object at 
$L / (16 m n)$.
\end{itemize}
Note that all valuations are integers smaller than $L^2 + L$, which is indeed smaller than ${t_2}^{c_2}$.
For the multi-reduction we use the function  
$f_e(\log t_2) := L/(64 m n) = {t_1}^{c_1} / (64 (t_2-1))$. Note that $f_e(\log t_2)\in \Omega(\sqrt{t_2}^{6} / t_2) = \Omega({t_2}^2)$, which is indeed super-polynomial in $\log t_2$.

Suppose that there exists a partition of $D$ into $n$ subsets $(D_1,\ldots, D_n)$ with sums equal to $S$, with at most $s$ shared items.
We construct an allocation of the objects in our instance by giving each agent $i$ the set of usual objects  corresponding to the items in $D_i$ (including fractions), as well as the compensation object $t_1-1+i$.
The number of shared items remains unchanged,  at most $s$.

We now prove that the allocation is EF (and hence PROP) not only for the constructed instance, but also for any instance with infinity distance at most $L/(64 m n)$.
In all such instances, each agent values every usual object $o$ at least $L d_o + L/(16 m n) - L/(64 m n) > L d_o$  and his compensation object at least $3 L / (8 n) - L/(64 m n) > 2 L / (8 n)$, so values his own bundle at least $L S + 2 L/(8 n)$; and values the bundle of every other agent at most $S+L/(8n)$. Hence the allocation is EF and also PROP.

Conversely, suppose there is a PROP allocation (or an EF allocation, which is always PROP), with at most $s$ shared objects.
For each agent $i$, the sum of all object values is larger than $L \cdot n S$,
so PROP requires each agent to receive a value larger than $L S$. 

Denote by $D_i$ the set of items from $D$ (including fractions) corresponding to the usual objects given to agent $i$.

The value of each potential bundle comes from integer multiples of $L$ (--- the $L\cdot d_o$ value of usual objects), and fractions of $L$ (--- the $L/(16 m n)$ value of usual objects, and the values of compensation objects). The sum of all fractions is strictly less than $L/n$.
Therefore, to get a value of at least $L S$, each agent must receive usual objects with 
$\sum_{o\in D_i} L d_o > L S-L/n$,
which implies 
$\sum_{o\in D_i} d_o > S-1/n$ for each agent $i\in[n]$.
Since 
$\sum_{o\in D} = n S$,
this implies 
$\sum_{o\in D_j} d_o < S+1$ for each agent $j\in[n]$.
Therefore, each bundle $D_i$ with 
$\sum_{o\in D_i} d_o \neq S$
must have fractions of items (as all item values in $D$ are integers).
Therefore, we can move fractions of items from bundles $D_i$ with sum larger than $S$ to bundles $D_j$ with sum smaller than $S$, without increasing the number of shared items. 
By iteratively moving these fractions, we can construct in polynomial time a partition of $D$ with at most $s$ shared items, in which all bin sums equal $S$.
\Halmos \endproof

\section{Hardness of finding an LP solution with fewer nonzero variables}
\label{sec:lp-bounds}
As the number of sharings in an allocation is related to the number of nonzero variables in a solution of an LP, it could be very helpful to have an algorithm for finding a solution minimizing the number of nonzero variables in a given LP. 
This problem is known as \emph{Min-RVLS} --- Minimum Relevant Variables in a Linear System. 
Unfortunately, Min-RVLS is NP-hard, and 
 it is also hard to find multiplicative approximations to the minimum \citep{amaldi1998approximability}.
Moreover, even finding a very modest additive approximation is NP-hard:
\begin{theorem}
\label{thm:lp-hard}
Unless P=NP, there is no polynomial-time algorithm that accepts as input a linear program with $k$ constraints and decides whether it has a feasible solution with at most $k-1$ nonzero variables.
\end{theorem}
\proof{Proof}
The proof is by reduction from \textsc{Partition}. 
Suppose we are given $k$ numbers $a_1,\ldots,a_k$ whose total sum is $2 S$, and have to decide whether they can be partitioned into two subsets such that the sum in each subset is $s$. 
This problem can be solves using an LP with $2 k$ variables. For each  $i\in[2]$ and $j\in [k]$, the variable $x_{ij}$ determines what fraction of the number $a_j$ is in subset $i$: 
\begin{align*}
\\
&& 
x_{1j}+x_{2j} = 1 && \text{ for all } j\in[k]
\\
&& 
\sum_{j=1}^m x_{1j} a_j =  \sum_{j=1}^m x_{2j} a_j &&
\\
&& 
x_{ij} \geq  0 && \text{ for all } i\in[n], j\in[k]
\end{align*}
There are $k+1$ equational constraints, so there always exists a solution with at most $k+1$ nonzeros. In this solution, at most one number is ``cut'' between the sets. Indeed, it is easy to solve the partition problem if we are allowed to cut one number: just order the numbers on a line and cut the line into two parts with equal sum. 

Now, if we could solve the ``few nonzeros'' problem, then we could decide whether the above LP has a solution with at most $k$ nonzeros, which would imply a solution to the partition problem (in which no number is cut). 
\Halmos \endproof
\end{APPENDICES}


\ACKNOWLEDGMENT{This work was inspired by Achikam Bitan, Steven Brams and Shahar Dobzinski, who expressed their dissatisfaction with the current trend of approximate-fairness (SCADA conference, Weizmann Institute, Israel 2018). We are grateful to participants of De Aequa Divisione Workshop on Fair Division (LUISS, Rome, 2019), Workshop on Theoretical Aspects of Fairness (WTAF, Patras, 2019) and the rationality center game theory seminar (HUJI, Jerusalem, 2019) for their helpful comments.  Suggestions of Herve Moulin, Antonio Nicolo, Nisarg Shah and Rohit Vaish were especially useful.

Several members of the stack-exchange network 
provided helpful answers, in particular: 
D.W, Peter Taylor, 
	Gamow, Sasho Nikolov,
	Chao Xu, Mikhail Rudoy,
	xskxzr, 
	Philipp Christophel, Kevin Dalmeijer,
	Bodo Manthey,
	and Niklas Rieken.
	
	We are grateful to the anonymous referees of WTAF 2019 and SAGT 2024 for their helpful comments.
}

\bibliographystyle{informs2014}
\bibliography{main,smooth}
\end{document}